\documentclass[12pt,a4paper]{article}

\usepackage[T1]{fontenc}
\usepackage{indentfirst,mathrsfs}
\usepackage{amsfonts,amsmath,amssymb,amsthm}
\usepackage{latexsym,amscd,multirow}
\usepackage{algorithm}
\usepackage{float}
\usepackage{algpseudocode}
\usepackage{amsbsy}
\usepackage{xcolor}
\usepackage{diagbox}
\usepackage{colortbl}
\usepackage{multirow}
\usepackage{makecell}
\usepackage{float}
\usepackage{cases}
\usepackage{dsfont}
\usepackage{enumerate}
\usepackage{cite}
\usepackage{pifont}
\usepackage{graphicx}
\usepackage{mathtools}
\usepackage{graphicx}

\usepackage{tikz}
\usetikzlibrary{decorations,arrows}
\usetikzlibrary{decorations.markings}
\usetikzlibrary{decorations.pathreplacing}
\usetikzlibrary {datavisualization.formats.functions}

\usepackage{hyperref}
\hypersetup{hidelinks,
	backref=true,
	pagebackref=true,
	hyperindex=true,
	breaklinks=true,
	colorlinks=true,
	linkcolor=black!60,
	citecolor=black!50,
	urlcolor=blue,
	bookmarks=true,
	bookmarksopen=true
}
\usepackage{url}
\newtheorem{theorem}{Theorem}
\newtheorem{lemma}{Lemma}
\newtheorem{corollary}{Corollary}
\newtheorem{Open problem}{Open Problem}

\newtheorem{prop}{Proposition}

\newtheorem{definition}{Definition}
\newtheorem{example}{Example}
\newtheorem{remark}{Remark}

\newtheorem{conjecture}{Conjecture}

\newcounter{alg}
\newlength{\lefttab}
\newlength{\numberoffset}
\newcommand{\cn}{\color{black}}

\newcommand{\s}{\textbf{s}}
\newcommand{\sn}{\textbf{s}_n}

\newcommand{\sni}{\textbf{s}_n^{\infty}}
\newcommand{\n}{\lfloor \frac{n}{2} \rfloor}
\newcommand{\NN}{\lfloor\frac{3n}{4}\rfloor}
\newcommand{\nup}{\lceil \frac{n}{2} \rceil}

\newcommand{\obr}[2]{\overbrace{#1 \dots #1}^{#2}}

\begin{document}
	\title{
		The structure and enumeration of periodic binary sequences with high nonlinear complexity
	}

		\author{
		Qin Yuan
		\thanks{Q. Yuan and X. Zeng are with
			Key Laboratory of Intelligent Sensing System and Security (Hubei University),
			Ministry of Education, Hubei Key Laboratory of Applied Mathematics, Faculty of
			Mathematics and Statistics, Hubei University, Wuhan, 430062, China. Email:
			\href{mailto:yuanqin2020@aliyun.com}{yuanqin2020@aliyun.com}, \href{mailto:xiangyongzeng@aliyun.com}{xiangyongzeng@aliyun.com} },
		Chunlei Li
		\thanks{C. Li is with the Department of Informatics, University of Bergen, Bergen, N-5020, Norway.
			Email: \href{mailto:chunlei.li@uib.no}{chunlei.li@uib.no}},
		Xiangyong Zeng$^{\,*}$
	}

	\date{}
	\maketitle

\begin{abstract}

Nonlinear complexity, as an important measure for assessing the randomness of sequences, is defined as the length of the shortest feedback shift registers that can
generate a given sequence.
In this paper,
 the structure of $n$-periodic binary sequences with nonlinear complexity larger than or equal to $\NN$ is characterized.
Based on their structure,  an exact enumeration formula for the number of such periodic sequences
is determined.

{\small {\bf Keywords:}} Periodic sequence, nonlinear complexity, enumeration, randomness.
\end{abstract}


\section{Introduction}

Pseudorandom sequences have wide applications in cryptography, communication, and ranging \cite{Golomb2017,Rueppel1986}.
The quality of randomness or unpredictability plays a crucial role in pseudorandom sequences.
 Various complexity measures have been proposed to assess the randomness of sequences \cite{Jansen, Niederreiter1999, Niederreiter2003, Rizo, YanKe, Sunex}.
Among them, linear complexity has been extensively studied \cite{ DingLegendre, Ding, Nian, Tang2005, PKe2013, PKe2018}.
The linear complexity of a sequence is defined as the length of the shortest linear feedback shift registers (LFSRs) that can generate it \cite{MasseyShirlei1996_Crypto}.
By the Berlekamp-Massey algorithm \cite{Massey1969}, a periodic sequence with linear complexity
$L$ can be uniquely reconstructed from any of its $2L$-length subsequences.
To resist this attack,
sequences in cipher systems should exhibit high linear complexity.
However, a sequence with high linear complexity may also be generated by a much shorter FSR with a nonlinear function, known as a nonlinear feedback shift register (NFSR).
Similarly to the definition of linear complexity,
allowing feedback functions of FSRs to be taken arbitrarily, i.e. removing the restriction of linear feedback functions,  gives the  notion  of nonlinear complexity  \cite{JansenPhD}.
Notably, NFSRs have been the primary building blocks in the design of stream ciphers, such as Trivium \cite{Trivium}, Grain \cite{Hell2008} and its variant Grain128-AEAD that advanced to the final round of the NIST light-weight cryptography standardisation process \cite{LWC}.
On the other hand, the understanding of NFSRs from theoretical perspective remains largely under-developed
	\cite{Helleseth}.

Existing researches on nonlinear complexity of sequences include efficient calculations of nonlinear complexity of sequences, construction of sequences with maximum/large nonlinear complexity, enumeration of sequences with a given nonlinear complexity~\cite{Jansen,Erdmann,Rizomiliotis2005,LKK2, liang, HX, SZLH, yuan},
and the relation of nonlinear complexity and correlation measures \cite{LWin, ChenZChen, Chenzhix}.
In the calculation of nonlinear complexity,
Jansen and Boekee \cite{Jansen}  initially related nonlinear complexity to the maximum depth of a directed
acyclic word graph, which can be employed to determine the nonlinear complexity (profile) of a given  binary  sequence.
By exploiting the special structure of  the  associated linear equations,
Rizomiliotis et al.
first investigated quadratic complexity \cite{Rizo} and subsequently extended their approach to general nonlinear complexity \cite{Rizomiliotis2005}.
 Limniotis, Kolokotronis, and Kalouptsidis
   later explored the relationship  between nonlinear complexity and Lempel-Ziv complexity, thereby presenting a recursive algorithm that produces a
minimal nonlinear feedback shift register (NFSR) of a given sequence \cite{LKK2}.

The construction of sequences with high nonlinear complexity has been investigated~\cite{Yilin, liang, SZLH, Xiao2018, PR, yuan}.
Existing results are largely concerned with $n$-length binary sequences and $n$-periodic binary sequences for an arbitrary positive integer $n\geq 3$.
For $n$-length sequences,
Liang et al.  constructed binary sequences with nonlinear complexity not less than $n/2$  and  determined the  exact number  of such finite-length sequences based on
their structure  \cite{liang}.
 In addition, function field techniques \cite{JinL, XingC2003, XingLam}
  have been employed to construct finite-length sequences exhibiting high nonlinear complexity \cite{Castellanos2022, LXY, HX}.
For $n$-periodic sequences with maximum nonlinear complexity $n-1$ or near-maximum nonlinear complexity $n-2$, their structures were derived using recursive approaches, and their exact enumerations were determined based on the general structural results \cite{SZLH, Xiao2018}.
Very recently, by revealing a more explicit relation between finite-length sequences
and the corresponding periodic sequences,
all periodic sequences with any prescribed nonlinear complexity were  intensively studied \cite{yuan}.
However, the understanding of $n$-periodic sequences with nonlinear complexity less than $ n-2$ is still limited.

The enumeration and distribution of nonlinear complexity of binary sequences provide an important statistical view in the behaviour of nonlinear complexity.
Erdmann and Murphy calculated an approximate number of sequences in each nonlinear complexity class
and proposed an approximate probability distribution for the nonlinear complexity
 by a statistical approach \cite{Erdmann}.
For periodic sequences,
Petrides and Mykkeltveit \cite{PM2006,Petrides2008} introduced recursive structures related to nonlinear complexity and  a composition operation on recursions, offering  an interesting  approach to classifying periodic binary sequences  with respect to their nonlinear complexity.
Nevertheless, 
 several cases are not resolved, and  the result cannot  be used to obtain the exact number of periodic sequences.
The exact distribution of nonlinear complexity for periodic sequences remains open.

 For $n$-periodic binary sequences, the structure and exact enumeration have been only studied for those sequences with nonlinear complexity $n-1$ and $n-2$ \cite{SZLH, Xiao2018}.
According to our recent work in \cite{yuan}, a set of  finite-length representative  sequences can be used to generate all such periodic sequences; however, the explicit structure of periodic sequences was not investigated in \cite{yuan}.
We are therefore motivated to further investigate the structure and enumeration of $n$-periodic binary sequences with high nonlinear complexity.
In this paper, we further
 prove that finite-length sequences used to generate shift inequivalent $n$-periodic sequences with nonlinear complexity larger than or equal to $ \NN$ must be  the unique representative sequences.
Building upon  these unique representative sequences,
the structure of binary $n$-periodic sequences with nonlinear complexity larger than or equal to $ \NN$ is derived,
and an exact enumeration formula for such sequences is  determined.
These results provide both the structure and the exact distribution of nonlinear complexity for such periodic sequences,
 which validates the approximated statistical results in \cite{Erdmann}.

The remainder of this paper is organized as follows.  Section \ref{Sec2} introduces basic notations, definitions and lemmas related to nonlinear complexity.
Section \ref{Sec3} is dedicated to determining the  representative sequences
from two distinct perspectives.
In Section \ref{Sec5_1}, the structure of $n$-periodic sequences with nonlinear complexity larger than or equal to $ \NN$ is determined.
In Section \ref{Sec5_2}, an exact enumeration formula for such periodic sequences  is derived.
Finally, Section \ref{Sec6} concludes the work of this paper.
In addition, Appendices A-C give technical and lengthy proofs of some results in the work for better readability.

\section{Preliminaries}\label{Sec2}
In this section, we shall recall some basics of the nonlinear complexity of sequences introduced by Jansen and Boekee \cite{Jansen, JansenPhD},
 and auxiliary results
on binary periodic sequences with high  nonlinear complexity  recently obtained in \cite{yuan}.

\begin{definition} (\!\!\cite{Jansen,JansenPhD}) \label{nlc}
	The nonlinear complexity of a sequence $\textbf{s}$ over an alphabet $\mathcal{A}$, denoted by $nlc(\textbf{s})$, is the length of the shortest feedback shift registers that can generate the sequence $\textbf{s}$.
\end{definition}

For a sequence $\s=(s_0,s_1,\dots)$, the term $s_{i+k}$ is deemed as
the \textit{successor} of the subsequence $\s_{[i:i+k]}=(s_i, \dots, s_{i+k-1})$ for certain
positive integers $i$ and $k$.
Some properties of the nonlinear complexity of sequences are recalled below.

\begin{lemma}(\!\!\cite{Jansen,JansenPhD})\label{lem}
	The nonlinear complexity of a sequence $\textbf{s}$ equals one plus the length of its longest identical subsequences that occur at least twice with different successors.
\end{lemma}	

A sequence $(s_{0},s_{1},\cdots, s_{n-1})$ is called a \textit{periodic} finite-length sequence if it is formed by multiple
repetitions of a short sequence of length $e$, where $e$ is a proper divisor of $n$; otherwise, it is termed an \textit{aperiodic} finite-length sequence, or simply said to be \textit{aperiodic}.
Throughout what follows, $\sn$  denotes  an aperiodic finite-length sequence, and $\sni$ denotes the infinite-length periodic sequence $\underbrace{\sn\dots\sn}_{\text{infinite}}$.

We will restrict our discussion to binary sequences over the integer ring  $\mathbb{Z}_2$.

\begin{lemma}(\!\!\cite{liang})\label{lem unique}
	For	a binary finite-length sequence
	$(s_0,s_1, \dots, s_{n-1})$,
	if it has nonlinear complexity  $c \geq \frac{n}{2}$, then
%
%
		there exists exactly one pair of identical subsequences of length $c -1$ with different successors in $(s_{0}, s_1, \dots, s_{n-1})$.
		
\end{lemma}




 In the following,  we will recall several definitions and lemmas from \cite{yuan}, which laid a foundation for discussions in the subsequent sections.


\begin{definition}
		\label{BB}
	For $c \geq \n$ and $1\leq d\leq \min\{n-c,\n\}$, we denote by $\mathcal{B}(n,c,d)$ the set of aperiodic binary sequences $\sn$ starting with $\s_{c+d}$,  defined  as follows:
	\begin{equation}\label{structure of finite}
		\mathcal{B}(n,c,d)=\{\,\textbf{s}_{n}=\textbf{s}_{c+d}\,\textbf{s}_{[c+d:n]}
		=\underbrace{((s_0, \dots, s_{d-1})^{q} \, (s_{0},\dots, s_{r-1}, \overline{s}_{r})}_{\text{length}=c+d}\, \textbf{s}_{[c+d:n]})\, \}
	\end{equation}
	where
	$q = \lfloor \frac{c+d-1}{d} \rfloor$,  $0 \leq r=(c+d-1)-qd<d$, $\s_d=(s_0, \dots, s_{d-1})$ is aperiodic, $(s_0, \dots, s_{d-1})^q=\obr{\s_d}{q \text{ repetitions}}$,
	$\overline{s}_{r}=s_{r}\oplus 1 $, and the subsequence $\textbf{s}_{[c+d:n]}$
	is chosen from $\mathbb{Z}_2^{n-c-d}$ arbitrarily.
	 The parameter $d$ is called the spacing of $\sn$. We define the set
	\begin{equation*}
		\mathcal{B}(n,c)=\bigcup\limits^{\min\{n-c,\n\}}_{d=1}\mathcal{B}(n,c,d).
	\end{equation*}
\end{definition}

Note that since $c \geq \n$, the structure  in \eqref{structure of finite}
ensures that the sequence $\sn$ is aperiodic, regardless of the choice of its subsequence $\textbf{s}_{[c+d:n]}\in \mathbb{Z}_2^{n-c-d}$.
To determine the nonlinear complexity of periodic sequences $\sni$ with $\sn\in  \mathcal{B}(n,c)$,
the added terms of $\textbf{s}_{n}$ was introduced in \cite{yuan}.

\begin{definition}
	\label{def-t}
	Given a certain positive integer $t$,
	if a sequence $\textbf{s}_{n}$ in $\mathcal{B}(n,c, d)$ with $c \geq \n$
	satisfies
	\begin{equation*}
		s_{n-1-i}=s_{(d-1-i)\,\text{mod}\,\,d}\,\, \text{ for } 0\leq i< t, \text{ and } s_{n-1-t}\neq  s_{(d-1-t)\,\text{mod}\,\,d},
	\end{equation*}
	then we call $s_{n-t}, \dots, s_{n-1}$ the added terms of $\textbf{s}_{n}$
	and denote by $add(\textbf{s}_{n})$ the number $t$ of the added terms of $\sn$.
\end{definition}

Define the left circular shift operators $L^i(\sn)=(s_i,s_{i+1},\dots,s_{n-1},$
$s_0,\dots,s_{i-1})$
and the right circular shift operators $R^i(\sn)=(s_{n-i},\dots,$
$s_{n-1},s_0,\dots,s_{n-i-1})$ for $i\geq 1$.
Lemma \ref{c-c+1} shows
the varying behavior of $nlc(\sn)$ and $add(\sn)$  under circular shift operators.

\begin{lemma}(\!\!\cite{yuan})\label{c-c+1}
	\noindent{\rm (i)} For $\sn\in\mathcal{B}(n,c, d)$ and a positive integer $k< c$, its left shift sequence
	$L^k(\sn)$ belongs to $\mathcal{B}(n,c-t,d)$.
	In particular, when $c-k \geq \n$, we have
	$$nlc(L^k(\sn))=nlc(\sn)-k=c-k.$$
		
		\noindent{\rm (ii)}	For $\textbf{s}_{n}\in  \mathcal{B}(n,c,d)$ with $c \geq \n$
	and $add(\textbf{s}_{n})=t$,
	we have 	$t < n-c$;  and
	for $1\leq k\leq \min\{t, n-c-d\}$, we have
		$$R^{k}(\textbf{s}_{n}) \in \mathcal{B}(n,c+k,d),  \,\,\, nlc(R^{k}(\textbf{s}_{n}))=c+k, \,\,\, \text{and \, }add(R^{k}(\textbf{s}_{n}))=t-k.$$	
\end{lemma}

%

\begin{definition}
	\label{ES_Def}
	Let $\sn$ be a sequence in $\mathcal{B}(n,c)$ and $E(\sn)=\{R^k(\sn): 0\leq k<n\} \cap \mathcal{B}(n,c)$.
	A sequence $\widetilde{\s}_n \in E(\sn)$ satisfying
	\begin{equation*}
		add(\widetilde{\s}_n) \geq add(\textbf{a}_n),\, \forall\, \textbf{a}_n\in E(\sn)
	\end{equation*} is said to be a representative sequence of $\sn$. Furthermore, we denote by $\mathcal{R}(n,c)$ the set of all sequence representatives in $\mathcal{B}(n,c)$,
	i.e.,
	\begin{equation}\label{Eq-Rnc}
		\mathcal{R}(n,c)= \bigcup_{\sn\in \mathcal{B}(n,c)} \big\{\widetilde{\s}_n\in E(\sn):add(\widetilde{\s}_n) \geq add(\textbf{a}_n),   \forall\, \textbf{a}_n \in E(\sn) \big\}.
	\end{equation}
\end{definition}

Readers may refer to \cite[Example 1]{yuan} for a better understanding of the notations in Definitions \ref{BB}-\ref{ES_Def}.

\begin{lemma}(\!\!\cite{yuan})\label{c+t}
	For $\textbf{s}_{n}$ in $\mathcal{R}(n,c)$ with $c\geq \n$,  one has $nlc(\textbf{s}_n^{\infty})=nlc(\sn)+add(\sn)$.
\end{lemma}

Lemma \ref{c+t} indicates that if we can determine a sequence representative of $\sn$, then $nlc(\textbf{s}_n^{\infty})$ can be calculated directly from
$nlc({\s}_n)$ and $add({\s}_n)$. It is to be noted that characterizing representative(s) of a sequence $\sn$ is not trivial since one needs to know the varying behavior of added terms
of all sequences in $E(\sn)$.
Some efforts in \cite{yuan} enabled the generation of all  $n$-periodic sequences with nonlinear complexity
larger than or equal to
$ \frac{n}{2}$, as  described below.


\begin{lemma}(\!\!\cite{yuan})
	\label{thm_core}
	Let $n$ and $\omega$ be two positive integers with $\omega\geq \frac{n}{2}$. Let
	 $\widetilde{\mathcal{P}}(n,\omega)$ be the set of binary sequences with period $n$ and nonlinear complexity $\omega$, and
	 $\mathcal{R}(n,\left\lceil\frac{n}{2}\right\rceil)$ be defined by \eqref{Eq-Rnc}.
	Then,
\begin{equation*}
\begin{array}{c}
		\widetilde{\mathcal{P}}(n,\omega)
= \left\{(R^k(\sn))^\infty:   \sn \in \mathcal{R}(n,\left\lceil\frac{n}{2}\right\rceil), \, add(\sn)=\omega-\left\lceil\frac{n}{2}\right\rceil, 0\leq k<n \right\}.
\end{array}
\end{equation*}
\end{lemma}

Although Lemma \ref{thm_core}  provides a method for generating binary periodic sequences using finite-length representative sequences,  it does not determine the exact form of the corresponding representative sequences in $\mathcal{R}(n,\left\lceil\frac{n}{2}\right\rceil)$.
To further characterize the precise structure of periodic sequences , it is essential to first investigate the properties and construction of their representative sequences.
According to Lemma \ref{c-c+1} (ii), we have \cn
\begin{equation}\label{eq_RE11}
	\begin{array}{cll}
		\vspace{0.2cm}
		& \{ R^k(\sn) :  \sn \in \mathcal{R}(n,\nup), \, add(\sn)=\omega-\nup, 0\leq k<n  \} \\
		\vspace{0.2cm}
		=& \{ R^k(\sn) :  \sn \in  \bigcup\limits_{1\leq d \leq n- \omega} \mathcal{R}(n,\nup,d), \, add(\sn)=\omega-\nup, 0\leq k<n  \}\bigcup\\
		\vspace{0.2cm}
		&  \,   \{ R^k(\sn) :  \sn \in  \bigcup\limits_{n- \omega< d \leq \n} \mathcal{R}(n,\nup,d), \, add(\sn)=\omega-\nup, 0\leq k<n   \} \\
		=& \bigcup\limits_{1\leq d < n- \omega}\{  R^k(\sn)  :  \sn \in \mathcal{R}(n,\omega,d),  add(\sn)=0, 0\leq k<n  \} \\
		& \bigcup   \bigcup\limits_{n- \omega \leq d \leq \n } \{ R^k(\sn) :  \sn \in \mathcal{R}(n,n-d,d),  \, add(\sn)=\omega+d-n, 0\leq k<n  \}.
	\end{array}
\end{equation}
Therefore,
in order to determine the structure of $	\widetilde{\mathcal{P}}(n,\omega)$ and  its number,
 we shall first characterize the sequence representatives in $\mathcal{R}(n,n-d,d)$ with the maximum
spacing and  $\mathcal{R}(n,\omega,d)$ with high nonlinear complexity $\omega$, respectively, in
Section \ref{Sec3}.

	\section{Characterization of   sequence representatives}\label{Sec3}

In this section, we examine the conditions under which a finite-length sequence can serve as a representative sequence from two distinct perspectives.
Through this analysis, we establish a necessary and sufficient condition for representative sequences with the maximum spacing in Subsection \ref{Sec3_1}, and further propose a criterion based on nonlinear complexity for identifying representative sequences in Subsection \ref{Sec4}.

The following lemma  will be used to determine sequence representatives  in
this section.
\begin{lemma}\label{b range}
	Let $c\geq \n$
	and $\textbf{s}_{n} \in \mathcal{B}(n, c, d_1)$ with $add(\sn)=t_1$.
	Suppose $\textbf{s}_{n}$ has a shift equivalent sequence $R^{h}(\textbf{s}_{n})\in \mathcal{B}(n,c, d_2)$
	with $add(R^{h}(\textbf{s}_{n}))=t_2$. Let $b = (n-c-d_1)-h$. Then
	$ t_{1} \leq b < d_{1}+d_{2}-t_{2}$.
\end{lemma}

The technical proof of Lemma 	\ref{b range}  	is given in Appendix A.

	\subsection{Sequence representatives with the maximum spacing}\label{Sec3_1}

 Let $c\geq \nup$.
 To
 determine the structure of periodic sequences,  this subsection  will present a necessary  and sufficient condition
such that a sequence  $\sn$ in $\mathcal{B}(n, c, d)$ with the maximum spacing $d=n-c$ can be a
 sequence representative.
 Moreover, the nonlinear complexity of the corresponding periodic sequence $\textbf{s}_n^{\infty}$
 can be directly determined.
 The main result in this subsection is presented below, followed by auxiliary results used in its proof.

\begin{prop}\label{t2>t1 seq} 			
	Let  $c\geq \lceil \frac{2n-1}{3} \rceil $  and $\textbf{s}_{n}$
	be a sequence
	in $\mathcal{B}(n, c, d)$ with 	$d=n-c$.		
	Then
	$\textbf{s}_{n} \in \mathcal{R}(n,c)$ if and only if ${\textbf{s}}_{[d:n]}$
	can't be expressed as ${\textbf{s}}_{[d:n]}=(s_{0},s_{1},\dots,s_{b-1})^{l}$ for any aperiodic subsequence
	$(s_{0},\dots,s_{b-1})$, where $b>1$ a proper divisor of $c$. 	
\end{prop}

To prove Proposition~\ref{t2>t1 seq}, we begin by characterizing the spacing
$d$ of the shifted sequence of $\sn$ with more added terms.

\begin{lemma}\label{d_2=b}
	Let  $c \geq \nup$  and
	$\textbf{s}_{n}$ be a sequence in $\mathcal{B}(n, c, d_1)$ with  $d_1=n-c$, $add(\sn)=t_1$.
	Suppose for certain integer $0 <b<n$, the shifted sequence $R^b(\sn)\in \mathcal{B}(n,c,d_2)$ with  $add(R^b(\sn))=t_2 \geq t_1$. If $d_1+d_2  \leq c+t_2+t_1 +1$,
	then $b=d_2$.   	
\end{lemma} \cn

The proof of Lemma \ref{d_2=b} is provided in Appendix B.
The condition of $d_1+d_2 \leq c+t_2+t_1+1$ in Lemma~\ref{d_2=b} can be transformed to a condition on nonlinear complexity.

\begin{corollary}\label{coro44cc}
	Let $c\geq \lceil \frac{2n-1}{3} \rceil $ and
	a sequence $\textbf{s}_{n} \in \mathcal{B}(n, c, d_1)$ with $d_1=n-c$ and $add(\sn)=t_1$.
	If there exists a positive integer $0 <b<n$ such that $R^b(\sn)\in \mathcal{B}(n,c,d_2)$ satisfies $add(R^b(\sn))=t_2\geq  t_1$, then $b=d_2$.	
\end{corollary}
\begin{proof}
Since $c\geq \lceil \frac{2n-1}{3} \rceil $, we have
$d_1+d_2 \leq 2(n-c) \leq c+1 \leq c+t_2+t_1+1.$
Thus, the result follows directly  from Lemma \ref{d_2=b}.
\end{proof}

	Based on the above preparations,
 Proposition~\ref{t2>t1 seq} is proved below.

\noindent\textbf{Proof of Proposition~\ref{t2>t1 seq}.}	
	According to Definition \ref{ES_Def},
	it suffices to prove that
	$\textbf{s}_{n}$ has a shift equivalent sequence $R^{b}(\textbf{s}_{n})$ in $\mathcal{B}(n,c)$ with $add(R^{b}(\textbf{s}_{n}))>add(\textbf{s}_{n})$ if and only if
	${\textbf{s}}_{[d:n]}$ can be expressed as ${\textbf{s}}_{[d:n]}=(s_{0},\dots,s_{b-1})^{l}$, where $l\geq 2$ and $(s_{0},\dots,s_{b-1})$ is aperiodic.

	For sufficiency, 	
	since ${\textbf{s}}_{[d:n]}=(s_{d},\dots,s_{n-b},\dots,s_{n-1})=(s_{0},s_{1},\dots,s_{b-1})^{l}$ with $b<n-d$, one can get  $(s_{n-b},\dots,s_{n-1})=(s_{0},\dots,s_{b-1})$.
	Moreover, from $\textbf{s}_{n}\in\mathcal{B}(n, c, d)$,
	$${\textbf{s}}_{c}=((s_{0},s_{1},\dots,s_{b-1})^{l-1}(s_{0},s_{1},\dots,s_{b-1}\oplus 1)).$$
	Thus, the shift equivalent sequence $R^{b}(\textbf{s}_{n})$ can be expressed as follows,
	\begin{equation*}
		\begin{array}{cll}
			R^{b}(\textbf{s}_{n}) &=((s_{n-b},\dots,s_{n-1})(s_{0},\dots,s_{c-1})(s_{c},\dots,s_{n-b-1}))\\
			&=((s_{0},\dots,s_{b-1})(s_{0},s_{1},\dots,s_{b-1})^{l-1}(s_{0},s_{1},\dots,s_{b-1}\oplus 1)(s_{c},\dots,s_{n-b-1}))\\
			&=((s_{0},s_{1},\dots,s_{b-1})^{l}(s_{0},s_{1},\dots,s_{b-1}\oplus 1)(s_{c},\dots,s_{n-b-1})).
		\end{array}
	\end{equation*}
	Since $(l-1)b+b=c$, it is clear that $R^{b}(\textbf{s}_{n})$ belongs to the set $\mathcal{B}(n,c,b)$ with $b<n-d=\min\{n-c,c\}$.
	Then it remains to prove that $add(R^{b}(\textbf{s}_{n}))>add(\textbf{s}_{n})$.
	For $\textbf{s}_{n}=(s_{0},\dots,s_{d-1})(s_{0},s_{1},\dots,s_{b-1})^{l}$, let $add(\textbf{s}_{n})=t_{1}$, we have $t_{1}\leq n-c-1\leq d$ by Lemma \ref{c-c+1} (ii).
	According to Definition \ref{def-t}
	and $t_{1}\leq d$, it follows that
	\begin{equation}\label{add-sn}
		\left\{\begin{array}{rll}
			s_{(b-1-i)\, \rm{mod}\,\,b}&=&s_{d-1-i},\,\, 0\leq i< t_{1}, \\
			s_{(b-1-t_{1})\, \rm{mod}\,\,b}& \neq  &s_{d-1-t_{1}}.
		\end{array}
		\right.
	\end{equation}
	Since $bl=n-d\geq d\geq t_{1}$, the added terms of $\textbf{s}_{n}$ is a subsequence of length $t_{1}$ in the end of $(s_{0},s_{1},\dots,s_{b-1})^{l}$.
	Based on
	${\textbf{s}}_{[d:n]}=(s_{0},s_{1},\dots,s_{b-1})^{l}$, $R^{b}(\textbf{s}_{n})$ can also be represented in the following form,
	\begin{equation*}
		\begin{array}{cll}
			R^{b}(\textbf{s}_{n}) &=((s_{n-b},\dots,s_{n-1})(s_{0},\dots,s_{d-1})(s_{d},\dots,s_{n-b-1}))\\
			&=((s_{0},\dots,s_{b-1})(s_{0},\dots,s_{d-1})(s_{0},\dots,s_{b-1})^{l-1}).
		\end{array}
	\end{equation*}
	Thus, it follows from Definition \ref{def-t} and (\ref{add-sn}) that
	\begin{equation}\label{t2t2t2}
		add(R^{b}(\textbf{s}_{n}))=(l-1)b+add(\textbf{s}_{n}).
	\end{equation}
	Hence, $R^{b}(\textbf{s}_{n}) \in \mathcal{B}(n,c,b)$ and $add(R^{b}(\textbf{s}_{n}))>add(\textbf{s}_{n})$ by $l\geq 2$.

	For necessity,
	since $\textbf{s}_{n}$ have shift equivalent sequences
	$R^{b}(\textbf{s}_{n})\in\mathcal{B}(n,c)$ with $add(R^{b}(\textbf{s}_{n}))$\\
	$>add(\textbf{s}_{n})$,
	from $c\geq \lceil \frac{2n-1}{3} \rceil $ and Corollary \ref{coro44cc},
	we have $R^{b}(\textbf{s}_{n})\in\mathcal{B}(n,c,b)$.	
	Let  $\textbf{v}_{n}=R^{b}(\textbf{s}_{n})$, then $s_i=v_{i+b}$
	and $v_{i_1}=v_{i_1+b}$ with $0\leq i_1 \leq c$.
	Thus $(v_0,v_1,\dots,v_{b-1})=(v_b,v_{b+1},\dots,v_{2b-1})=(s_0,s_1,\dots,s_{b-1})$.
	It implies that $(s_{0},\dots,s_{b-1})$ is aperiodic and
	\begin{equation}\label{v_c+b}
		\textbf{v}_{c+b}=(s_0,s_1,\dots,s_{b-1})^q(s_0,\dots,s_{r-1},\overline{s}_{r}).
	\end{equation}
	Due to $\textbf{s}_{n}\in\mathcal{B}(n,c,d)$, we have
	$(s_{c-b},\dots,s_{c-2},s_{c-1})=(s_{n-b},\dots,s_{n-2},\overline{s}_{n-1})$,
	that is to say,
	$$(v_{c},\dots,v_{c+b-2},v_{c+b-1})=(v_{0},\dots,v_{b-2},\overline{v}_{b-1}).$$
	Recall that  $(v_0,v_1,\dots,v_{b-1})=(s_0,s_1,\dots,s_{b-1})$,
	thus it is evident that
	$$(v_{c},\dots,v_{c+b-2},v_{c+b-1})=(s_0,s_1,\dots,\overline{s}_{b-1}).$$
	Together with \eqref{v_c+b},
	it can be seen that
	$(s_{r+1},\dots,{s}_{b-1})(s_0,\dots,s_{r-1},\overline{s}_{r})=(s_0,s_1,\dots,\overline{s}_{b-1})$,
	which implies $r=b-1$ since all shifts of a aperiodic sequence $\s_b$ are different.
	Therefore, $\textbf{v}_{c+b}=(s_0,s_1,\dots,s_{b-1})^{q}(s_0,s_1,\dots,\overline{s}_{b-1})$.
	So the first $c$ terms of $R^{b}(\textbf{s}_{n})$ provide that
	$$(s_{n-b},\dots,s_{n-1},s_{0},\dots,s_{c-b-1})=(s_{0},\dots,s_{b-1})^{l}.$$
	And it follows from $s_i=s_{i+d}$, $0 \leq i \leq c-2$ that
	$$
	((s_{0},s_{1},\dots,s_{c-b-1})(s_{n-b},\dots,s_{n-1}))
	=((s_{d},s_{d+1},\dots,s_{n-b-1})(s_{n-b},\dots,s_{n-1}))
	={\textbf{s}}_{[d:n]}.$$
	Combining the above two equations implies ${\textbf{s}}_{[d:n]}=(s_{0},s_{1},\dots,s_{b-1})^{l}$ with $n-d=c= bl$.
	 The desired conclusion thus follows. 	
\hfill $\square$


Following the proof for the necessity in Proposition \ref{t2>t1 seq}, we immediately have the following
corollary.

\begin{corollary}\label{myloveadd}
	Let $c \geq \nup$ and $\textbf{s}_{n} \in \mathcal{B}(n, c, d)$ with $d=n-c$.
	If its shifted sequence $R^{b}(\textbf{s}_{n}) \in \mathcal{B}(n, c, b)$, then
	
	\noindent{\rm (i)} when $add(R^{b}(\textbf{s}_{n})) \geq add(\textbf{s}_{n})$,   we have ${\textbf{s}}_{[d:n]}=(s_{0},s_{1},\dots,s_{b-1})^{l}$, where
	$(s_{0},\dots,s_{b-1})$ is a certain  aperiodic subsequence and the integer   $l \geq 2$.
	
	\noindent{\rm (ii)} when $add(R^{b}(\textbf{s}_{n})) > add(\textbf{s}_{n})$,  we have $1\leq b < d$, $b\,|\,n-d$ and if $d\neq \frac{n}{2}$, then $b\nmid d$.
\end{corollary}
\begin{proof}
	(i) It follows from  the necessity proof for Proposition \ref{t2>t1 seq}.
	
	(ii)  By (i), we have ${\textbf{s}}_{[d:n]}=(s_{0},\dots,s_{b-1})^{l}$, thus it follows $b\,|\,n-d$.
Since $R^{b}(\textbf{s}_{n})\in\mathcal{B}(n,c,b)$, we get $1\leq b \leq d=n-c$.
	When $d\neq \frac{n}{2}$,
	$\textbf{s}_{[d:2d]}$ is a subsequence of ${\textbf{s}}_{[d:n]}=(s_{0},\dots,s_{b-1})^{l}$.
	Hence, if $b\,|\, d$, then $\textbf{s}_{d}=\textbf{s}_{[d:2d]}=(s_{0},\dots,s_{b-1})^{m}$ with $d=bm$.
From (\ref{structure of finite}), we know that $\textbf{s}_{d}$ is an aperiodic sequence, a contradiction.
Thus, if $d\neq \frac{n}{2}$, then $b\nmid d$ and $b<d$.
When $d=  \frac{n}{2}$, if $b=d$, then it is clear that $add(R^b(\textbf{s}_{n}))=add(\sn)$, a contradiction.
Hence, $b \neq d$.
Therefore, $1\leq b < d$, $b\,|\,n-d$ and if $d\neq \frac{n}{2}$, then $b\nmid d$.
\end{proof}

According to Proposition \ref{t2>t1 seq}, if there exists an integer $b$ with $1<b<c$ such that ${\textbf{s}}_{[d:n]}=(s_{0},\dots,s_{b-1})^{l}$ for an aperiodic subsequence $(s_{0},\dots,s_{b-1})$, then such an integer
$b$ is unique, thereby implying that the sequence $R^{b}(\textbf{s}_{n})$  is also unique.	
	Now we are ready to give the nonlinear complexity of the corresponding
	periodic sequence $\textbf{s}_n^{\infty}$.

					\begin{corollary}\label{cor2}
						Let  $c\geq \lceil \frac{2n-1}{3} \rceil$.
	For $\textbf{s}_{n} \in \mathcal{B}(n, c, d)$ with  $d=n-c$, 					
 if ${\textbf{s}}_{[d:n]}=(s_{0},s_{1},\dots,s_{b-1})^{l}$ with an aperiodic subsequence $(s_{0},s_{1},\dots,s_{b-1})$ and an integer $l\geq 2$,
		then $R^{b}(\textbf{s}_{n})$ is the sequence representative and $nlc(\textbf{s}_n^{\infty})=c+n-d-b+add(\textbf{s}_{n})$;
		otherwise, $\textbf{s}_{n}$ is the sequence representative and $nlc(\textbf{s}_n^{\infty})=c+add(\textbf{s}_{n})$.		
	\end{corollary}
	\begin{proof}
 According to Proposition \ref{t2>t1 seq},  $\textbf{s}_{n} \in \mathcal{B}(n, c,d )$ with $d=n-c$ has either only one shift equivalent sequence $R^{b}(\textbf{s}_{n})$ with $add(R^{b}(\textbf{s}_{n}))$
		$>add(\textbf{s}_{n})$ or no such shift equivalent sequence.
		When there exists an integer $b$ satisfying ${\textbf{s}}_{[d:n]}=(s_{0},\dots,s_{b-1})^{l}$,
		$R^{b}(\textbf{s}_{n})$ is the sequence representative of $ES(\textbf{s}_{n})$.
		By (\ref{t2t2t2}), we have $add(R^{b}(\textbf{s}_{n}))=(l-1)b+add(\textbf{s}_{n})=
		n-d-b+add(\textbf{s}_{n})$.
		Due to Lemma \ref{c+t}, we have
		$nlc(\textbf{s}_n^{\infty})=c+add(R^{b}(\textbf{s}_{n}))=	n+c-d-b+add(\textbf{s}_{n})$.
		When there is not a positive integer $b$ satisfying ${\textbf{s}}_{[d:n]}=(s_{0},\dots,s_{b-1})^{l}$,
		it follows from Proposition \ref{t2>t1 seq} that $\textbf{s}_{n}$ is the sequence representative of $ES(\textbf{s}_{n})$, which implies $nlc(\textbf{s}_n^{\infty})=c+add(\textbf{s}_{n})$ by Lemma \ref{c+t}. The proof is finished.
	\end{proof}

From experimental data, we observe that Corollary \ref{coro44cc} also holds with a less restrictive condition on the nonlinear complexity $c$.
Nevertheless, it cannot be proved by the same technique. We provide a conjecture below. Note that if the conjecture holds, then statements in Proposition \ref{t2>t1 seq} and Corollary \ref{cor2} can be similarly proved for sequences $\textbf{s}_{n}\in \mathcal{B}(n, c, n-c)$ with $c\geq \nup$.

	\begin{conjecture}
		For $c\geq \nup$ and  $\textbf{s}_{n}\in \mathcal{B}(n, c, n-c)$,
	if there exists an integer $0 <b<n$ such that
	$R^b(\sn)\in \mathcal{B}(n,c,d_2)$ satisfies $add(R^b(\sn))> add(\sn)$,
	then $b=d_2$.
	\end{conjecture}

	\subsection{Determining sequence representatives with high nonlinear complexity}\label{Sec4}
	For $\sn$ in $\mathcal{B}(n,c)$ with $c\geq \n$,  the set $E(\sn) = \{L^k(\sn)\,:\, 0\leq k<n \}\cap \mathcal{B}(n,c)$
	may contain several sequences.
	In this section, we shall establish  a lower bound $c_{0}$ such that a sequence $\sn$ with nonlinear complexity $c\geq c_0$
	is a sequence representative, which allows us to
	determine the structure of periodic sequences
	and calculate the nonlinear complexity of $\sni$ directly as $nlc(\sni)=nlc(\sn)+add(\sn)$.

	\begin{theorem}\label{thm_lowerbound}
		Let $n\geq 1$ and $k$ be integers.
		For a sequence $\sn \in \mathcal{B}(n, c)$,
		if $c \geq c_0=\lfloor\frac{3n}{4}\rfloor$, then
		$\sn \in \mathcal{R}(n, c)$ and the nonlinear complexity of $\sni$ satisfies $nlc(\textbf{s}_n^{\infty})=nlc(\sn)+add(\sn)$.
		Moreover, the tight bound $c_0$ for even $n$ is
		$
		c_0= \begin{cases}
			\lfloor\frac{3n}{4}\rfloor- 1, & \text{ if } n= 8k, \,{ \forall} k,\\
			\lfloor\frac{3n}{4}\rfloor, & \text{ if } n = 8k+2,\,8k+6,\,{ \forall} k,\\
			\lfloor\frac{3n}{4}\rfloor- 2,& \text{ if } n = 8k+4, k \geq 2.\\
		\end{cases}
		$	
		%
		%
		\end{theorem}
		\begin{proof}
			We first transform the cyclic shift equivalent sequences in $\mathcal{B}(n,c)$ into the shift equivalent sequences in $\mathcal{B}(n,\n)$ with some restrictions on the parameters of sequences. Then
			we shall use the properties of shift equivalent sequences in $\mathcal{B}(n,\n)$, given by Lemma \ref{b range}, to prove the first statement $\sn \in \mathcal{R}(n, c)$ .

			When $c \geq \lfloor\frac{3n}{4}\rfloor$,	
		suppose there exists a pair of shift equivalent sequences $(\textbf{s}'_n,\textbf{v}'_n)$ in $\mathcal{B}(n, c)$
		with
		$0 \leq add(\textbf{s}'_n)<add(\textbf{v}'_n)$.
		Then
		let  $\textbf{s}'_n \in \mathcal{B}(n, c,d_1)$ and
		$\textbf{v}'_n\in \mathcal{B}(n, c,d_2)$.
		Since $c>\frac{n}{2}$, we have $d_1 \leq n-c$ and $d_2 \leq n-c$.
		If $d_1=n-c$, then from $c \geq \lfloor\frac{3n}{4}\rfloor  \geq   \lceil \frac{2n-1}{3} \rceil$
		and Corollary \ref{coro44cc}, we get $b=d_2$.
		Thus by Corollary \ref{myloveadd} (ii), we know that $d_2=b <d_1= n-c$. Therefore,
		$d_1+d_2 < 2(n-c)$.
	
	Let $m=\n$.	It follows from Lemma \ref{c-c+1} (i) that $\sn=L^{c-m}(\textbf{s}'_n)\in \mathcal{B}(n, m,d_1)$ and
		$\textbf{v}_n=L^{c-m}(\textbf{v}'_n)\in \mathcal{B}(n, m,d_2)$.
		So 	$(\sn,\textbf{v}_n)$ is a pair of shift equivalent sequences in $\mathcal{B}(n, m)$ with $t_1=add(\sn) =add(\textbf{s}'_n)+(c-m) \geq c-m$ and $t_2=add(\textbf{v}_n) \geq t_1+1$.
		Suppose $\textbf{v}_n=R^{a+b}(\sn)$ with $a=n-m-d_1$.
		According to Lemma \ref{b range}, it remains to deal with the case of $t_{1}\leq b<d_{1}+d_{2}-t_{2}$.	
			Recall that we have
			\begin{equation}\label{threecon}
				d_1 +d_2 < 2(n-c), \ \,	t_1\geq c-m, \ \, t_2 \geq t_1+1\geq (c-m )+1.
			\end{equation}
			Since $c \geq c_0=\lfloor\frac{3n}{4}\rfloor$ and $m=\n$,
			it implies $d_{1}+d_{2}\leq  2\lceil \frac{n}{4} \rceil  -1 \leq  2\lfloor \frac{n}{4} \rfloor  +1 \leq  t_{1}+t_{2}$,
			which contradicts $t_{1}<d_{1}+d_{2}-t_{2}$.	
			Therefore there does not exist a pair of shift equivalent sequences $(\textbf{s}'_n, \textbf{v}'_n)$ in $\mathcal{B}(n, c)$	satisfying $add(\textbf{s}'_n)<add(\textbf{v}'_n)$ when $c\geq \lfloor\frac{3n}{4}\rfloor$.
			 That is to say, 	each sequence  $\sn \in \mathcal{B}(n, c)$ with $c \geq c_0=\lfloor\frac{3n}{4}\rfloor$
			is a sequence representative.
			The first statement follows.

			In the subsequent analysis, to establish the tight bound for even $n$, according to the form of $n$, we partition the discussion into four distinct cases:	$n=8k+2, 	n=8k+6, n=8k+4$, and $n=8k$.
					To improve the readability of the proof, we first outline the key ideas and main steps before presenting the detailed proof.

			For the cases
				(i) $n=8k+2$ and
				(ii) $n=8k+6$, a lower bound
			$c_0 =\lfloor\frac{3n}{4}\rfloor  $ has  been established, and all sequences
			with nonlinear complexity not less than $c_0$
			have been shown to be representative sequences.
			To verify the tightness of this bound, it suffices to provide a counterexample: a sequence with nonlinear complexity equal to
		$c_0   -1$ that is not a representative sequence. This confirms that the bound cannot be further improved. The detailed proofs are provided below.

			{\bf	(i)} When $n=8k+2$, suppose a sequence has the following form
			$$
			\sn=(s_0, \dots, s_{8k+1}) = ((\underline{\alpha \beta^{2k-1}}\, \alpha \beta^{2k-1} \, \alpha \beta^{2k-1} \, \beta) \alpha \beta ^{2k}),
			$$ where $\beta = \overline{\alpha}$ and $\beta^l$ is the sequence given by $l$ repetitions of $\beta$ for a positive integer $l$. It is clear that
			$\sn$ belongs to $\mathcal{B}(8k+2,4k+1, 2k)$  with $add(\sn)=2k-1$, where $\mathbf{s}_{2k}$ is underlined.
			Consider the sequences
			$$\begin{array}{l}
				\mathbf{u}_n=R^{2k-1}(\sn) = ((\underline{\beta ^{2k-1}\alpha}\, \beta^{2k-1}\alpha \,  \beta^{2k-1}\alpha \,  \beta^{2k-1}\beta) \alpha\beta), \\
				\mathbf{v}_n=R^{6k}(\sn) = (( \underline{\beta^{2k-2}\alpha\beta^2} \, \beta^{2k-2}\alpha\beta^2 \, \beta^{2k-2}\alpha\beta^2 \,  \beta^{2k-3}\alpha)\beta).
			\end{array}
			$$
			It is clear that $\mathbf{u}_n \in \mathcal{B}(8k+2, 6k, 2k)$ with $add(\mathbf{u}_n) = 0$ and that
			$\mathbf{v}_n \in \mathcal{B}(8k+2, 6k, 2k+1)$ with $add(\mathbf{v}_n) = 1$.
			That is to say, for the sequence $\mathbf{u}_n$ in $\mathcal{B}(8k+2, 6k)$, its cyclic shift sequence $\mathbf{v}_n=R^{4k+1}(\mathbf{u}_n)$ belongs to $\mathcal{B}(8k+2, 6k)$ with $add(\mathbf{v}_n)>add(\mathbf{u}_n)$.
			It indicates that when $n=8k+2$, the lower bound $c_0=6k+1=\lfloor\frac{3n}{4}\rfloor$ is a tight bound such that any sequence $\sn\in \mathcal{B}(n, c)$ with $c\geq c_0$ is a sequence representative.
			
			{\bf	(ii)}	When $n=8k+6$, the proof can be obtained directly by replacing $2k$ to $2k+1$ in the proof of $n=8k+2$.

			For the cases
			(iii)  $n=8k+4$ and
			(iv) $n=8k$,
			we focus on the case $n = 8k + 4$, as the case $n = 8k$ can be handled in a similar manner.	
			In the first step, we assume that there exists a finite sequence with nonlinear complexity at least $c_0 =\lfloor\frac{3n}{4}\rfloor -2 $ that is not a representative sequence. Based on the relationships among parameters, all possible subcases are classified and analyzed. In each case, two distinct representations of an aperiodic subsequence $\s_{d_1}$ or $\mathbf{v}_{d_2}$ are derived, and it is shown that they have different Hamming weights, which leads to a contradiction.
			In the second step, we verify the tightness of the bound $c_0 = \lfloor \frac{3n}{4} \rfloor - 2$ by constructing a sequence with nonlinear complexity exactly equal to $c_0 - 1$ that is not a representative sequence.
	The complete proof is deferred to Appendix B.
					\end{proof}

		\begin{remark}\label{remark_2}
			Theorem \ref{thm_lowerbound} indicates that if $c \geq c_0$ then $\sn \in \mathcal{B}(n, c)$ is a sequence representative, that is $\sn \in \mathcal{R}(n, c)$.
			As shown in Theorem \ref{thm_lowerbound}, the tight lower bounds for even $n$ are given.  In fact, as for $\mathcal{B}(n,c)$ with odd $n$, the tight bound of nonlinear complexity $c_0$  is very close to the tight bounds for even numbers $n-1$ and $n+1$.
			And if $nlc(\sn)=c$ is known,
			then the value of $add(\sn)$  can be computed in  $O(n)$ time, thus
			$nlc(\sni)$ can be determined directly by  $c+ add(\sn)$.
			Compared to directly computing the nonlinear complexity of periodic sequences, it is more efficient.
			When $n$ is even, the tight lower bound provides the exact range within which this method for computing the nonlinear complexity of periodic sequences is applicable.	
		
			Furthermore, for any sequence $\textbf{a}_n$ with nonlinear complexity $c\geq c_0$, by Lemma \ref{lem unique} it contains one pair of identical subsequences of length $(c-1)$ with different successors. This implies $L^i(\textbf{a}_n)=\textbf{s}_{n} \in \mathcal{B}(n, c)$ for certain $i$.
			Therefore,
			for each finite-length sequence $\textbf{a}_n$ with nonlinear complexity $c$,
			if $c\geq c_0$, then the nonlinear complexity of the corresponding periodic sequence can be directly obtained, namely,
			$nlc(\textbf{a}_n^{\infty})=c+add(L^i(\textbf{a}_{n}))$.
		\end{remark}
		Below we present an example to illustrate the result of Theorem \ref{thm_lowerbound}.
		\begin{example}
			Take an example for $\s_{20}=(\underline{1000110}10001101000\mathbf{10}) \in \mathcal{B}(20,13,7)$ with $add(\s_{20})=2$, where $\s_d$ is underlined and added terms are in bold.
			When $n=20$, it follows from Theorem \ref{thm_lowerbound} that the tight bound $c_0$ is equal to 13.
		Due to  $nlc(\s_{20})=13= c_0$ and $add(\s_{20})=2$,  by Theorem \ref{thm_lowerbound}, we can directly determine the nonlinear complexity of its corresponding periodic sequence: $$nlc(\s_{20}^{\infty})=nlc(\s_{20})+add(\s_{20})=13+2=15.$$
		\end{example}


\section{Structure of periodic sequences  with high nonlinear complexity }\label{Sec5_1}

In this section, by combining the results from Section \ref{Sec3},
 we prove that when $\omega \geq \NN$, each sequence $\sn\in \mathcal{B}(n, \left\lceil\frac{n}{2}\right\rceil )$ with $add(\sn)=\omega- \left\lceil\frac{n}{2}\right\rceil$
is the unique sequence representative.
This, in turn, enables us to
characterize the structure of $n$-periodic sequences with
$\omega \geq \NN$.
We first present the main result on the structure of such periodic sequences.

\begin{theorem}\label{seq_structure}
	A binary $n$-periodic sequence has nonlinear complexity $\omega \geq \NN$
	if and only if one of its aperiodic subsequences, denoted as $\sn=(s_0,s_1,\dots ,s_{n-1})$, can be represented
	as one of the following two forms:
	
	\noindent{\rm (i)}  when $1 \leq d \leq n-\omega-1$,
	\begin{equation}\label{stru1}
		\sn = ( \underbrace{(s_0, \dots, s_{d-1})^{q} \, (s_{0},\dots, s_{r-1}, \overline{s}_{r})}_{\text{length}=\omega +d}\, (s_{\omega +d},\dots,s_{n-2} )\, \overline{s}_{d-1} )
	\end{equation}
	where
	$q = \lfloor \frac{\omega+d-1}{d} \rfloor$,  $0 \leq r=(\omega+d-1)-qd<d$, $\s_d=(s_0, \dots, s_{d-1})$ is aperiodic, and the subsequence $(s_{\omega +d},\dots,s_{n-2})$
	is chosen from $\mathbb{Z}_2^{n-\omega-d-1}$ arbitrarily.
	
	\noindent{\rm (ii)}  when $n-\omega \leq d \leq \n$,
	$$\sn= (s_0,s_1,\dots,s_{n-1}) = {(s_0, \dots, s_{d-1})^{q'} \, (s_{0},\dots, s_{r'-1}, \overline{s}_{r'})}$$
	where $q' = \lfloor \frac{n-1}{d} \rfloor$,  $0 \leq r'=(n-1)-q'd<d$,
	and
	the aperiodic subsequence $\s_d=(s_0,s_1,\dots,s_{d-1})$ is chosen such that
	if $t = \omega+d-n=0$, then  ${s}_{(n-1) \,\text{mod}\,\,d} = {s}_{d-1}$;
	if $t = \omega+d-n > 0$, then
	$$\overline{s}_{(n-1) \,\text{mod}\,\,d} = {s}_{d-1}, {s}_{(n-i) \,\text{mod}\,\,d} = {s}_{(d-i)\,\text{mod}\,\,d}, \,\,2 \leq i  \leq t, 	 {s}_{(n-t-1) \,\text{mod}\,\,d} = \overline{s}_{(d-t-1)\,\text{mod}\,\,d}. $$
\end{theorem}

The proof of Theorem \ref{seq_structure} relies on the following necessary lemmas.
To begin, based on Lemma \ref{thm_core}, we consider the following set
\begin{equation}\label{cru_set}
	\mathcal{S}_\omega = \left\{ \sn : \sn\in \mathcal{B} \left(n, \left\lceil\frac{n}{2}\right\rceil \right), add(\sn)=\omega- \left\lceil\frac{n}{2}\right\rceil \right\}.
\end{equation}
First, we show that every sequence in $\mathcal{S}_\omega$ is a  representative sequence as described in Lemma  \ref{pop2},
Then, we prove in Lemma~\ref{pop1} that  the sequences in $\mathcal{S}_\omega$ are shift inequivalent.

Note that according to \eqref{eq_RE11},
all sequences in $\mathcal{S}_\omega$ together with their shifted versions can be partitioned into two disjoint subsets as follows:
\begin{equation}\label{eq_RE}
	\begin{array}{cll}
	 	\vspace{0.2cm}
		& \{ R^k(\sn) :  \sn \in \mathcal{B}(n,\nup), \, add(\sn)=\omega-\nup, 0\leq k<n  \} \\
		\vspace{0.2cm}
		=& \{  R^k(\sn)  :  \sn \in \mathcal{S}_A, 0\leq k<n  \} \cup   \{ R^k(\sn) :  \sn \in \mathcal{S}_B, 0\leq k<n  \},
	\end{array}
\end{equation}
where
$$
\begin{array}{cll}
	\mathcal{S}_A& = &  \bigcup\limits_{1\leq d < n- \omega}\{    \sn :\sn \in\mathcal{B}(n,\omega,d),  add(\sn)=0  \}, \\
	\mathcal{S}_B & = &  \bigcup\limits_{n- \omega \leq d \leq \n } \{  \sn :  \sn \in \mathcal{B}(n,n-d,d),  \, add(\sn)=\omega+d-n   \}.
\end{array}
$$

For $ \sn \in \mathcal{B}(n,c)$,
 recall that its shift equivalence class $E(\sn)$ is defined in Definition \ref{ES_Def} as
$E(\sn)=\{R^k(\sn): 0\leq k<n\} \cap \mathcal{B}(n,c)$.

\begin{lemma}\label{pop2}
Let $\omega \geq \NN$. If $\sn \in \mathcal{B}(n, \left\lceil\frac{n}{2}\right\rceil )$ with $add(\sn)=\omega- \left\lceil\frac{n}{2}\right\rceil$, then $\sn$ is a sequence representative  of its shift equivalence class $E(\sn)$.
\end{lemma}	
\begin{proof}	
	From \eqref{eq_RE}, it suffices to prove that the sequences in  $\mathcal{S}_A$ and
$\mathcal{S}_B$	are all representative sequences.
	When $1\leq d < n- \omega$, since $\omega \geq \NN$, it follows from Theorem \ref{thm_lowerbound} that every sequence in	$\mathcal{B}(n, \omega)$ is a representative sequence. Therefore, all sequences in
$\mathcal{S}_A$ are representative sequences.

	When $n- \omega\leq d \leq \n$, sequences are in the set $\mathcal{S}_B$.
	Suppose that $\sn\in \mathcal{S}_B$
	is not a representative sequence. Then, for $\sn\in \mathcal{B}(n,c,d_1) $ with
$c=n-d_1$ and $add(\sn)=t_1=\omega+d_1-n $, there exists a shifted sequence such that $R^b(\sn)\in \mathcal{B}(n,c,d_2) $ satisfying $add(R^b(\sn))>t_1 $.	
	Note that
	$$c+t_1+t_2+1 \geq  d_1+2\omega+1-n \geq d_1+\left(
	2\left\lfloor\frac{3n}{4}\right\rfloor +1-n  \right)
	\geq d_1+\left\lfloor\frac{n}{2}\right\rfloor \geq d_1+d_2.  $$
	Thus, by Lemma \ref{d_2=b} (i), we have $b=d_2$.
 Furthermore, based on Corollary \ref{myloveadd}, it follows that
	${\textbf{s}}_{[d:n]}=(s_{0},s_{1},\dots,s_{b-1})^{l}$,
where $(s_{0},\dots,s_{b-1})$ is certain aperiodic sequence and the integer $l \geq 2$.
 Consequently,
		$$s_i=s_{i-d \,(\text{mod} \,\,b)}, \,\,\, d\leq i\leq n-1.$$

Due to $\sn \in \mathcal{B}(n,n-d,d)$, we have
$s_i=s_{i \text{ mod } d}, d\leq i \leq n-2$ and $s_{n-1}=\overline{s}_{(n-1) \text{ mod } d}$.
Since $\s_{d}=\s_{[d:2d]}$ is contained in $\s_{[d:n]}= (s_0,s_1,\dots,s_{b-1})^l$,   it follows that $s_i=s_{i \text{ mod } b}$, $0\leq i<d$.
Hence, we obtain that
\begin{equation}\label{constraints1}
	\left\{\begin{array}{cll}
		s_{(n-1-d) \text{ mod } b}&= s_{n-1} =&\overline{s}_{(n-1) \text{ mod } d}=\overline{s}_{((n-1) \text{ mod } d )\text{ mod } b}, \\
		s_{(n-2-d) \text{ mod } b}&=s_{n-2} =&s_{(n-2) \text{ mod } d}={s}_{((n-2) \text{ mod } d )\text{ mod } b}, \\
		...&=&...\\
		s_{(d+1) \text{ mod } b}&=s_{2d+1} =&s_{1} \qquad  \qquad  =s_{1}, \\
		s_{d \text{ mod } b}&=s_{2d} \ \ \,=&{s}_{0}  \qquad  \qquad ={s}_{0}, \\
	\end{array}\right.  \ \
\end{equation}	
where for each line in \eqref{constraints1}, the first equality follows from 	$s_i=s_{(i-d) \text{ mod } b}, d\leq i\leq n-1$, the second equality follows from 	$s_i=s_{i \text{ mod } d}, d\leq i \leq n-2$ and $s_{n-1}=\overline{s}_{(n-1) \text{ mod } d}$, and the third equality
holds since $s_i=s_{i \text{ mod } b}$, $0\leq i<d$.

Moreover let $t=add(\sn)=\omega+d-n$, then $t<d$.
Thus according to Definition \ref{def-t}, since $\s_{[d:n]} = (s_0,s_1,\dots,s_{b-1})^l$ and
$s_i=s_{i \text{ mod } b}$, $0\leq i<d$,
we have
\begin{equation}\label{constraints2}
	\left\{\begin{array}{cll}
		s_{(d-1) \text{ mod } b}=s_{d-1}&=s_{n-1}=&s_{b-1}, \\
		s_{(d-2) \text{ mod } b}=s_{d-2}&=s_{n-2}=&s_{b-2}, \\
		...&=&...\\
		s_{(d-t) \text{ mod } b}=s_{d-t}&=s_{n-t}=&s_{(b-t) \text{ mod } b}. \\
	\end{array}\right.  \ \
\end{equation}

Since  $n- \omega\leq d \leq \n$ and $ \omega \geq \NN$, we have $b+d\leq \omega$.
Thus, the combined number of equations from \eqref{constraints1} and \eqref{constraints2} is
$(n-1-d)-(d-t)+1=\omega-d \geq b$ in total.
Thus $(s_{(n-1-d) \text{ mod } b},s_{(n-2-d) \text{ mod } b},\dots,	s_{(n-b-d) \text{ mod } b})$ as a shifted version of $\s_b=(s_0,\dots,s_{b-1})$, has Hamming weight $wt(\s_b)$.
Again by equations in \eqref{constraints1} and \eqref{constraints2}, it follows
$$
\begin{array}{cll}
	&(s_{(n-1-d) \text{ mod } b},s_{(n-2-d) \text{ mod } b},\dots,	s_{(n-b-d) \text{ mod } b})\\
	=&(\overline{s}_{((n-1) \text{ mod } d )\text{ mod } b},s_{((n-2) \text{ mod } d )\text{ mod } b},\dots,	s_{((n-b) \text{ mod } d )\text{ mod } b}).
\end{array}
$$
It implies that $(\overline{s}_{((n-1) \text{ mod } d )\text{ mod } b},s_{((n-2) \text{ mod } d )\text{ mod } b},\dots,	s_{((n-b) \text{ mod } d )\text{ mod } b})$
has Hamming weight $wt(\s_b)$, a contradiction.	
Therefore, all sequences in $\mathcal{S}_B$ are representative sequences.
The proof is finished.
\end{proof}

From Lemma \ref{pop2}, we have
\begin{align}
	\mathcal{S}_\omega =  \,&\left\{ \sn :  \sn \in \mathcal{B}\left(  n, \left\lceil\frac{n}{2}\right\rceil   \right), \, add(\sn)=\omega- \left\lceil\frac{n}{2}\right\rceil  \right\} \notag \\
	= \, &\left\{ \sn :  \sn \in \mathcal{R}\left(n, \left\lceil\frac{n}{2}\right\rceil \right), \, add(\sn)=\omega- \left\lceil\frac{n}{2}\right\rceil   \right\}.    \label{lem_R_B}
\end{align}
We now prove that the sequences in $\mathcal{S}_\omega$ are shift inequivalent, implying that each representative sequence is unique.

\begin{lemma}\label{pop1}
	Let $n \geq 3$ and $\omega \geq \NN$.
All representative sequences in the set $\{\sn \in \mathcal{R}(n,\nup), \, add(\sn)=\omega-\nup  \}$ are shift inequivalent.
\end{lemma}
\begin{proof}
According to \eqref{lem_R_B}, it suffices to prove that the sequences in
$$\left\{\sn : \sn \in \mathcal{B}\left(n,  \left\lceil\frac{n}{2}\right\rceil \right), \, \, \, add(\sn)=\omega-\left\lceil\frac{n}{2}\right\rceil  \right\}$$
  are shift inequivalent.
By \eqref{eq_RE}, it is equivalent to proving that
$$
\begin{cases}
(1) \text{ sequences in the set }  \mathcal{S}_A   \text{ are shift inequivalent};\\
(2) \text{ sequences in the  set } \mathcal{S}_B   \text{ are shift inequivalent}; \\
(3) \text{ sequences in } \mathcal{S}_A  \text{ are shift inequivalent to those in }  \mathcal{S}_B.		
\end{cases}
$$

	{\bf Case (1)}: Suppose that there exist shift equivalent sequences in $ \mathcal{S}_A$, namely
	$\sn\in\mathcal{B}(n, \omega, d_1)$ and $R^h(\sn)=\textbf{v}_n\in \mathcal{B}(n, \omega, d_2)$,
then we have $add(\sn)=add(\textbf{v}_n) =0$, $1 \leq d_1,d_2 \leq n-\omega -1 \leq \lceil \frac{n}{4} \rceil$.
Let $\textbf{s}'_n=L^{\omega-\n}(\sn)$ and $\textbf{v}'_n=L^{\omega-\n}(\textbf{v}_n)$. Thus we can see that
\begin{equation}\label{praran}
	\left\{\begin{array}{cll}
			&\textbf{s}'_n\in \mathcal{B}(n, \n,d_1), & add(\textbf{s}'_n)=t_1=\omega-\n, \\
			R^h(\textbf{s}'_n)=&\textbf{v}'_n\in \mathcal{B}(n, \n,d_2), & add(\textbf{v}'_n)=t_2=\omega-\n,
		\end{array}\right.
	\end{equation}
	where $ 1 \leq d_1,d_2 \leq n -\omega -1 $.

	According to Lemma \ref{b range}, we have
	\begin{equation}\label{hhh_ran}
		\left(n-    \left\lfloor\frac{n}{2}\right\rfloor   -d_1 \right)+t_{1} \leq h < \left(n-   \left\lfloor\frac{n}{2}\right\rfloor   -d_1  \right)+(d_{1}+d_{2}-t_{2}).
	\end{equation}
Thus $n-h   \leq \n+d_1-t_1 <\n+d_1$.
It implies that
	$\textbf{v}'_{[h-t_1:n]}=\s'_{[-t_1:n-h]}$
	is contained in $\s'_{[-t_1:\n+d_1-1]}$.
	From the structure of $\textbf{s}'_n$, it follows that
	each $d_1$-length sequence in $\textbf{v}'_{[h-t_1:n]}=(v'_{h-t_1},v'_{h-t_1+1},\dots,v'_{n-1})$				
	is a shifted version of $\s'_{d_1}$.				
	Moreover from  \eqref{praran} and \eqref{hhh_ran}, we have
	 $$h-t_1 \leq n- \left\lfloor \frac{n}{2} \right\rfloor +d_2-t_2-1-t_1
	\leq n-\left\lfloor \frac{n}{2} \right\rfloor +(n-\omega -1)-1- \left(2\omega-2 \left\lfloor \frac{n}{2} \right\rfloor  \right)
	\leq \left\lfloor \frac{n}{2} \right\rfloor -d_1.$$
	thus $(v'_{\n-d_1},v'_{\n-d_1+1},\dots,v'_{\n-1})$ and $(v'_{\n+d_2-d_1},v'_{\n+d_2-d_1+1},\dots,v'_{\n+d_2-1})$
	are contained in $\textbf{v}'_{[h-t_1:n]}$,
	implying that they both are shifted sequences of $\s'_{d_1}$.			
	While it follows from the structure of $\textbf{v}'_n$ in \eqref{structure of finite} that $$(v'_{\n-d_1},v'_{\n-d_1+1},\dots,v'_{\n-2})=(v'_{\n+d_2-d_1},v'_{\n+d_2-d_1+1},\dots,v'_{\n+d_2-2})$$
	and $v'_{\n-1}\neq v'_{\n+d_2-1}$, a contradiction.		
	Therefore,	all sequences in $ \mathcal{S}_A$ are shift inequivalent.

{\bf Case (2)}:  Suppose that there are shift equivalent sequences in $\mathcal{S}_B$, then
\begin{equation*}
	\left\{\begin{array}{rll}
		\textbf{s}_{n}\in \mathcal{B}(n,n-d_1,d_1),   & n-\omega \leq d_1\leq \n,     &  add(\textbf{s}_{n})=\omega+d_1-n , \\
		\textbf{v}'_{n}\in \mathcal{B}(n,n-d_2,d_2), &    n-\omega \leq d_2 \leq d_1  \leq \n,     & add(\textbf{v}'_{n})=\omega+d_2-n.
	\end{array}
	\right.
\end{equation*}
Let $c=n-d_1$, $\textbf{v}_{n}=L^{n-d_2-c}(\textbf{v}'_{n})$, then
\begin{equation*}
	\left\{\begin{array}{rll}
		\textbf{s}_{n}\in \mathcal{B}(n,c,d_1),   & n-\omega \leq d_1\leq \n,     &  add(\textbf{s}_{n})=t_1=\omega+d_1-n , \\
		\textbf{v}_{n}\in \mathcal{B}(n,c,d_2), &    n-\omega \leq d_2 \leq d_1  \leq \n,     & add(\textbf{v}_{n})=t_2=\omega+d_1-n.
	\end{array}
	\right.
\end{equation*}
Because of $\textbf{s}_{n}$ and $\textbf{v}_{n}$ are shift equivalent,   we let $\textbf{v}_{n}=R^b(\s_n)$ with $0<b<n$, then $\textbf{s}_{n}=R^{n-b}(\textbf{v}_{n})$.
Note that
\begin{equation}\label{improlove}
	c+t_1+t_2+1 =  d_1+2\omega+1-n \geq d_1+\left(2  \left\lfloor\frac{3n}{4}\right\rfloor    +1-n  \right) \geq d_1+\left\lfloor\frac{n}{2}\right\rfloor \geq d_1+d_2.
\end{equation}
Hence, on the one hand, for the pair of sequences $(\textbf{s}_{n}, R^b(\s_n))$,
by Lemma \ref{d_2=b}  and \eqref{improlove}, we obtain $b=d_2$;
on the other hand, for $(  \textbf{v}_{n}, R^{n-b}(\textbf{v}_{n}))$,
again by  Lemma \ref{d_2=b}  and \eqref{improlove}, we get $n-b=d_1$.
Thus, $b=d_2=n-d_1$.
From $d_1,d_2 \leq \n$, we derive $b=d_1=d_2=\frac{n}{2}$.
This leads to $t_1=t_2=0$,
contradicting
$$t_1=t_2=\omega+d_1-n \geq \left\lfloor\frac{3n}{4}\right\rfloor  - \frac{n}{2} \geq 1, \,\, (n\geq 3).$$
	Therefore, sequences in $\mathcal{S}_B$ are shift inequivalent.

{\bf Case (3)}:  Suppose that certain sequences in $\mathcal{S}_A$ are shift equivalent to those in $\mathcal{S}_B$:
\begin{equation*}
	\left\{\begin{array}{rll}
		\textbf{s}_{n}\in \mathcal{B}(n,n-d_1,d_1),   & n-\omega \leq d_1\leq \n, &       add(\textbf{s}_{n})=t=\omega+d_1-n , \\
		R^{h}(\s_n)= \textbf{v}'_{n} \in\mathcal{B}(n,\omega,d_2), & 1 \leq d_2< n-\omega,  &     add(\textbf{v}'_{n})=0.
	\end{array}\right.
\end{equation*}
Let $c=n-d_1$, and $\textbf{v}_{n}=L^{t}(\textbf{v}'_{n}) $, then
\begin{equation*}
	\left\{\begin{array}{rll}
		\textbf{s}_{n}\in \mathcal{B}(n,c,d_1),   & n-\omega \leq d_1\leq \n, &       add(\textbf{s}_{n})=t_1=\omega+d_1-n , \\
		R^{h-t}(\s_n)=\textbf{v}_{n} \in\mathcal{B}(n,c,d_2), & 1 \leq d_2< n-\omega ,  &     add(\textbf{v}_{n})=t_2=\omega+d_1-n.
	\end{array}
	\right.
\end{equation*}
Since
$$ c+t_1+t_2+1 =  d_1+2\omega+1-n \geq d_1+ \left(2  \left\lfloor\frac{3n}{4}\right\rfloor    +1-n \right) \geq d_1+\left\lfloor\frac{n}{2}\right\rfloor \geq d_1+d_2,$$
it follows from Lemma \ref{d_2=b} that $h-t=d_2$.

Based on $\textbf{s}_{n}\in\mathcal{B}(n,n-d_1,d_1)$ and $add(\textbf{s}_{n})=t$, we have
\begin{equation}\label{idd_11}	
	s_{-t-1}=\overline{s}_{d_1-t-1},	\,\,s_{i}={s}_{i+d_1},    \,\,   \,\, i \in [-t, n-d_1-2]
	\,\, \mbox{ and  } \,\,
	s_{n-d_1-1}=\overline{s}_{n-1}.
\end{equation}
Since $ R^{d_2+t}(\s_n)= \textbf{v}_{n}\in\mathcal{B}(n,\omega,d_2)$ and $add( R^{d_2+t}(\s_n) )=0$, we derive
\begin{equation}\label{idd_22}
	\begin{array}{rll}
		& s_{-d_2-t-1}=\overline{s}_{-t-1},	\,\,	s_{i}={s}_{i+d_2},    \,\,  \,\, i \in [-d_2-t, \omega-d_2-t-2], \\
		&	 \mbox{ and } \,\, s_{\omega-d_2-t-1}=\overline{s}_{\omega-t-1}=\overline{s}_{n-d_1-1},
	\end{array}
\end{equation}
where all subscripts in the above equations are  taken modulo $n$.

From \eqref{idd_22}, we know that arbitrary consecutive $d_2$ terms in  $$(\overline{s}_{-d_2-t-1},{s}_{-d_2-t},s_{-d_2-t+1},\dots,{s}_{\omega-t-2},\overline{s}_{\omega-t-1})$$
consist of a shifted sequence of $\textbf{v}_{d_{2}}$.
Then $(\overline{s}_{-d_2-t-1},{s}_{-d_2-t},s_{-d_2-t+1},\dots,{s}_{-t-2})$
is a shifted sequence of  $\textbf{v}_{d_{2}}$.
By \eqref{idd_11} and \eqref{idd_22}, it follows that
\begin{align*}
	& (\overline{s}_{-d_2-t-1},{s}_{-d_2-t},s_{-d_2-t+1},\dots,{s}_{-t-2}) \\
	=\, & ({s}_{-t-1},{s}_{-t},s_{-t+1},\dots,{s}_{d_2-t-2})  \\
	=\, & (\overline{s}_{d_1-t-1},{s}_{d_1-t},s_{d_1-t+1},\dots,{s}_{d_1+d_2-t-2}).
\end{align*}
Thus,  $(\overline{s}_{d_1-t-1},{s}_{d_1-t},s_{d_1-t+1},\dots,{s}_{d_1+d_2-t-2})$
is a shifted version of  $\textbf{v}_{d_{2}}$.
Consider the subsequence $({s}_{d_1-t-1},{s}_{d_1-t},s_{d_1-t+1},\dots,{s}_{d_1+d_2-t-2})$, since its subscripts satisfy
$$-d_2-t \leq   d_1-t-1, \, d_1-t, \dots, d_1+d_2-t-2 \leq \omega-t-2,$$
it is also a shifted version of $\textbf{v}_{d_{2}}$, a contradiction.
Therefore, the sequences in $\mathcal{S}_A$ are not shift equivalent to the sequences in  $\mathcal{S}_B$.

Combining the above three subcases, there are no shift equivalent sequences in
$$\left\{\sn : \sn \in \mathcal{B} \left(n,  \left\lceil\frac{n}{2}\right\rceil  \right), \, \, \, add(\sn)=\omega-\left\lceil\frac{n}{2}\right\rceil  \right\}.$$
Thus from  \eqref{lem_R_B}, each sequence $\sn\in\mathcal{R}(n,\nup)$  with $add(\sn) = \omega - \nup$ is a unique representative sequence.
\end{proof}

Combining Lemmas \ref{pop2} and \ref{pop1}, we immediately obtain the following proposition.

\begin{prop}\label{prop_new}
	Each sequence $\s_n\in \mathcal{B}(n,\nup)$ with $add(\sn)\geq  \NN-\nup$ is the  unique  representative sequence in $E(\sn)$.
\end{prop}

Based on Proposition \ref{prop_new}, combined with Lemma \ref{c-c+1} and Theorem \ref{thm_lowerbound}, we can conclude from the nonlinear complexity of finite-length sequences that $\s_n$ is the unique representative sequence.

\begin{corollary}\label{cor_new}
	Let   integers  $k \geq  0$ and $c \geq  \NN$. For $\s_n \in \mathcal{B}(n,c)$, we have
	
	\noindent{\rm (i)} $\s_n$ is the unique representative sequence.

	\noindent{\rm (ii)}
	 when $n = 4k + 2$,
	the inequality $c \geq \NN$ gives a tight lower bound such that every sequence $\s_n \in \mathcal{B}(n, c)$ is a unique representative sequence.
\end{corollary}	
\begin{proof}
	(i)
	For any sequence $\s_n \in \bigcup\limits_{1 \leq d \leq n - c} \mathcal{B}(n, c, d)$, it follows from Lemma \ref{c-c+1} (i) that $L^{c - \nup}(\s_n) \in \mathcal{B}(n, \nup, d)$,
	and $add(L^{c - \nup}(\s_n)) = add(\s_n) + (c - \nup) \geq \NN - \nup$.
	Therefore,
	{\small
\begin{align*}
	&\bigg\{ \sni : \s_n \in \bigcup\limits_{1 \leq d \leq n-c} \mathcal{B}(n,c,d),\ c \geq \left\lfloor \frac{3n}{4} \right\rfloor \bigg\} \\
	\mathrel{\scalebox{1.5}{$\subset$}}   &\bigg\{ \sni : \s_n \in \bigcup\limits_{1 \leq d \leq \n} \mathcal{B}\left(n,\left\lceil \frac{n}{2} \right\rceil,d\right),\ add(\s_n) \geq \left\lfloor \frac{3n}{4} \right\rfloor - \left\lceil \frac{n}{2} \right\rceil \bigg\}.
\end{align*}}Combining this with Proposition~\ref{prop_new}, we conclude that when $c \geq \NN$,
$L^{c - \nup}(\s_n)$ is the unique representative sequence of its shift equivalence class
$$E(L^{c - \nup}(\s_n)) = \{ R^k(\s_n) : 0 \leq k < n \} \cap \mathcal{B}(n, \lceil n/2\rceil )$$
 within $\mathcal{B}\left(n, \left\lceil \frac{n}{2} \right\rceil \right)$.
Consequently, $\s_n$ is the unique representative sequence of its shift equivalence class
$E(\s_n) = \{ R^k(\s_n) : 0 \leq k < n \} \cap \mathcal{B}(n, c)$ within $\mathcal{B}(n, c)$.
That is, $\s_n$ is the unique representative sequence.
Therefore, the first statement follows.
	
	(ii) 	Based on Theorem~\ref{thm_lowerbound}, when $n = 4k + 2$ for any integer $k$,
	$c \geq \NN$ represents a tight lower bound that ensures $\s_n$ is a representative sequence.
	Therefore, $c \geq \NN$ also serves as a tight lower bound ensuring that $\s_n$ is the unique representative sequence.
\end{proof}

\begin{remark}
	Experimental data indicate that when $n = 4k + 2$ for any integer $k$,
	$add(\s_n) = t \geq \NN - \nup$ serves as a tight lower bound ensuring that $\s_n \in \mathcal{B}(n, \nup)$ is the unique representative sequence.
	For other values of $n$, when $n \in \{7, 11, 13, 15, \dots\}$,
	the bound $add(\s_n) \geq \NN - \nup$ in Proposition~\ref{prop_new} also represents a tight lower bound.
\end{remark}

With the preparations above,
we are now ready to prove the structure of $n$-periodic sequences with nonlinear complexity
$\geq \NN$ in Theorem \ref{seq_structure}.
 Define ${\mathcal{P}}(n,\omega)$
as the set of binary shift inequivalent sequences with period $n$ and nonlinear complexity $\omega$.

\noindent\textbf{Proof of Theorem \ref{seq_structure}.}
	For $\NN \leq  \omega\leq  n-1$,
	according to Lemma~\ref{thm_core} and Proposition~\ref{prop_new},  we have	
	$$
	\begin{array}{cll}
		{\mathcal{P}}(n,\omega)
		\vspace{1mm}
		&=&\{ \sni :  \sn \in \mathcal{B}(n,\nup), \, add(\sn)=\omega-\nup  \}.
	\end{array}			
	$$		
	Furthermore, from the equation in \eqref{eq_RE},  we get
	\begin{equation*}
		\begin{array}{cll}
			{\mathcal{P}}(n,\omega) 		
			=&\sum\limits_{1 \leq d < n- \omega } \{\sni :   \sn \in\mathcal{B}(n,\omega,d), \, add(\sn)=0 \}	\\
			& +\sum\limits_{n- \omega \leq d \leq \n}\{ \sni :  \sn \in  \mathcal{B}(n,n-d,d), \, add(\sn)=\omega+d-n  \}	.
		\end{array}	
	\end{equation*}

(i) When  $1 \leq d \leq n-\omega-1$,
based on the structure of $\sn \in\mathcal{B}(n,\omega,d)$ given in \eqref{structure of finite} and the definition of $add(\sn)$ in Definition \ref{def-t},
we have $s_{n-1} = \overline{s}_{(n-1) \,\text{mod}\,\,d}$ and
$\overline{s}_{n-1} = {s}_{d-1}$,
then ${s}_{(n-1) \,\text{mod}\,\,d}= {s}_{d-1}$, thus
$\s_n$ has the form as in \eqref{stru1}.

(ii) When $n-\omega \leq d \leq \n$, $\sn \in  \mathcal{B}(n,n-d,d)$
has the form of
$$\sn= (s_0,s_1,\dots,s_{n-1}) = {(s_0, \dots, s_{d-1})^{q'} \, (s_{0},\dots, s_{r'-1}, \overline{s}_{r'})}$$
where $q' = \lfloor \frac{n-1}{d} \rfloor$,  $0 \leq r'=(n-1)-q'd<d$.
Moreover,	if  $add(\sn)=0$, then $\sn$ satisfies
${s}_{n-1} = \overline{s}_{d-1}$;
if $add(\sn)=t = \omega+d-n > 0$, then according to the structure of $\sn$ and Definition \ref{def-t},  $\sn$ satisfies
\begin{equation*}
	\overline{s}_{(n-1) \,\text{mod}\,\,d} = {s}_{d-1}, {s}_{(n-i) \,\text{mod}\,\,d} = {s}_{(d-i)\,\text{mod}\,\,d}, \,\,2 \leq i  \leq t,	 {s}_{(n-t-1) \,\text{mod}\,\,d} = \overline{s}_{(d-t-1)\,\text{mod}\,\,d}.
\end{equation*}	
The desired statements (i) and (ii) thus follow.
\hfill $\square$

Theorem \ref{seq_structure} in this paper presents the structure of $n$-periodic sequences with nonlinear complexity
	$\geq \NN$.
 This further improves the result in \cite{yuan}, which provided a method to generate the set $	\widetilde{\mathcal{P}}(n,\omega)$, but didn't offer an insight into the internal structure of the periodic sequences in
$	\widetilde{\mathcal{P}}(n,\omega)$.

\section{Enumeration of periodic sequences with high nonlinear complexity}\label{Sec5_2}

This section is devoted to determining the exact value  of $|{\mathcal{P}}(n,\omega)|$ for $\omega \geq \NN$,  where	$|{\mathcal{P}}(n,\omega)|$  denotes the cardinality of ${\mathcal{P}}(n,\omega)$.

Let $A(d)$ be the set consisting of all aperiodic $d$-length sequences $\s_d $.
Since any sequence of length $n$
can be regarded as a periodic sequence with period
dividing $n$, we obtain
$\sum\limits_{d | n} |A(d)|	 =2^n$.
According to M$\rm{\ddot{o}}$bius Inversion Formula \cite{Lidl1997}, we have
$|A(d)|	= \sum\limits_{e | d} \mu(e) 2^{d/e}
$, where the M$\rm{\ddot{o}}$bius function defined by
\begin{equation}\label{MoMo}
	\mu(n) = \begin{cases}
	1 & \text{ if } n=1,\\
	(-1)^k & \text{ if } n \text{ is the product of } k \text{ distinct primes},\\
	0 & \text{ if } n \text{ is divisible by the square of a prime}.\\
	\end{cases}
\end{equation}
	According to the structure of sequences in $\mathcal{P}(n,\omega)$ give in  Theorem \ref{seq_structure},
		when $1 \leq d<n-\omega$,
	if $\s_d$ has been selected, then $(s_0,s_1,\dots,s_{d+\omega-1})$ and $s_{n-1}$ are determined,  only $(s_{d+\omega},s_{d+\omega+1},\dots,s_{n-2})$ can be chosen arbitrarily in $\mathbb{Z}_{2^{n-\omega-1-d}}$.
	Thus the number of sequences is
	$$	2^{n-\omega-1-d}|A(d)|	= \sum\limits_{e | d} \mu(e)  2^{ n-\omega-1-d+d/e}. $$
When $n-\omega \leq d \leq \n$,  if $t = \omega+d-n=0$, then $d = n-\omega$ and
denote
$${N}(n,d,0)	=
\{\s_{d} \in A(d) : s_{n-1 \,(\text{mod}\,\,d)}= s_{d-1} \},$$
where $A(d)$ is the set of all aperiodic $d$-length sequences $\s_d $.
 If $t = \omega+d-n>0$, then we denote
\begin{equation}\label{condiN}
	\begin{array}{cll}
		&{N}(n,d,t)\\
		=& \{\s_d \in A(d) :	\overline{s}_{(n-1) \,\text{mod}\,\,d} = {s}_{d-1}, {s}_{(n-i) \,\text{mod}\,\,d} = {s}_{(d-i)\,\text{mod}\,\,d}, \,\,2 \leq i  \leq t, \\
		&	 {s}_{(n-t-1) \,\text{mod}\,\,d} = \overline{s}_{(d-t-1)\,\text{mod}\,\,d} \}.
	\end{array}	
\end{equation}	
Thus
$$|{\mathcal{P}}(n,\omega)|	 = \sum\limits_{d =1}^{n- \omega-1}\sum\limits_{e | d} \mu(e) 2^{ n-\omega-1-d +d/e} + |	{N}(n,n-\omega,0)|	
 +\sum\limits_{n- \omega < d \leq \n}|	 {N}(n,d,\omega+d-n)|		.$$

	In what follows, we shall determine the values of $N(n, d, 0)$ and 	
	$N(n, d, t)$ with $t> 0$, respectively.

	\begin{lemma}\label{lemma fi}
		Let ${N}(n,d,0)  =\{ \s_d \in A(d): s_{(n-1) \,\text{mod}\,\,d} = s_{d-1} \}	$.
		For composite integers $d$, we have	
		\begin{equation*}
			|	{N}(n,d,0)|	=
			\begin{cases}
				|A(d)|	=\sum\limits_{e | d} \mu(e) 2^{d/e},  &\text{ if }  d \,|\,n,  \\
				2^{d-1}-2-	\sum\limits_{e|d, 1<e<d}|{N}(n,e,0)|, &   \text{ if }  d \nmid n.      \\
			\end{cases}		
		\end{equation*}
		Moreover, for prime $d$ we have
		$$	
	|{N}(n,d,0)|	=
		\begin{cases}
			|A(d)|	=2^{d}-2,  	&\text{ if }  d \,|\,n,  \\
			2^{d-1}-2,    	&\text{ if }  d \nmid n.       \\
		\end{cases}	
		$$
	\end{lemma}
	\begin{proof}
		If $ d \,|\,n$, then 		$|\{ \s_d \in A(d): s_{(n-1) \,\text{mod}\,\,d} = s_{d-1} \}|	=	|\{ \s_d \in A(d) \}|	=	|A(d)|	$.
		If $d \nmid n$, then ${N}(n,d,0)=\{\s_d \in A(d): s_{(n-1) \,\text{mod}\,\,d} = s_{d-1}  \}$. 	
		 Thus, to determine the value of ${N}(n,d,0)$,
		 in  $2^{d-1}$ vectors $(s_0,\dots,s_{d-2}) \in \mathbb{Z}_{2}^{d-1}$,
		 we only need to  exclude those sequences
		 that result in
		 $\s_{[0:d]}=(s_0,s_1,\dots,s_{d-2},s_{(n-1) \,\text{mod}\,\,d})$
		  not being  an  aperiodic finite-length sequence. 		
	Below, we will now study the case where
	$\s_{[0:d]}=(s_0,s_1,\dots,s_{d-1})=(s_0,s_1,\dots,s_{d-2},s_{n-1 \,(\text{mod}\,\,d)})$
 is   a periodic finite-length sequence.

		If $(s_0,s_1,\dots,s_{d-1})$ is a periodic finite-length sequence, then it equals $(\s_e)^{d/e}$ where $\s_e$ is aperiodic with $e|d, 1\leq e<d$.
		Thus it suffices to characterize $\s_e$ such that $(s_0,s_1,\dots,s_{d-1})$ is periodic.
		When $e=1$, $s_e$ has two choices, (0) or (1), which yields $\s_{[0:d]}=(0^d)$ or $(1^d)$.
		When $e>1$, since $(s_0,s_1,\dots,s_{d-1})=(\s_e)^{d/e}$, we have $s_{d-1} = s_{e-1}$ and $s_{(n-1) \,\text{mod}\,\,d}=s_{(n-1) \,\text{mod}\,\,e}$.
		Moreover by $ s_{(n-1) \,\text{mod}\,\,d}=s_{d-1} $,
		it implies
		$\s_e$ satisfies $s_{(n-1) \,\text{mod}\,\,e}=s_{e-1} $.
		Then the number of 	$\s_e$  is $|\{ \s_e \in A(e): s_{(n-1)\,\text{mod}\,\,e} = s_{e-1} \}|=|N(n,e,0)|$,
		where $A(e)$		is the set consisting of all aperiodic finite-length sequences $\s_e$
		of length $e$.
		Thus it follows $|	{N}(n,d,0)|	=2^{d-1}-2-	\sum\limits_{e|d, 1<e<d}|	{N}(n,e,0)|	$.
		
		Hence the calculation of $|	N(n,d,0)|	$ can be reduced to $|	N(n,e,0)|	$ with prime $e$.
		For prime $d$,
		$|	N(n,d,0)|	=|A(d)|$ is obvious when $d \,|\,n$.
		When  $d \nmid n$,
		since prime $d$ only has proper divisor 1, we have
		$|N(n,d,0)|=2^{d-1}-2$. The proof is thus complete.
	\end{proof}

	Note that the value of ${N}(n,d,0)$ is
	determined by $d$ and $(n \,\,\text{mod}\,\,d)$.
When $t>0$, 	
 the equation
$$\overline{s}_{(n-1) \,\text{mod}\,\,d} = {s}_{d-1}, {s}_{(n-i) \,\text{mod}\,\,d} = {s}_{(d-i)\,\text{mod}\,\,d}, \,\,2 \leq i  \leq t, 	 {s}_{(n-t-1) \,\text{mod}\,\,d} = \overline{s}_{(d-t-1)\,\text{mod}\,\,d} $$
 can be viewed equivalently as a system of binary linear equations
	${A}\cdot \mathbf{x}^T=\mathbf{b}^T$ where $\mathbf{x}=(s_0,s_1,\dots,s_{d-1})$, $\mathbf{b}=(1,\overbrace{0, \dots 0}^{t-1},1)$ and the matrix $A=(a_{i,j})_{(t+1)\times d}$ satisfies that $a_{i,((n-i)\,\,\text{mod}\,\,d)+1}=1$, $a_{i,(d-i)+1}=1$, $i,j \geq 1$ and other elements are 0.
	Thus ${N}(n,d,t)$ defined in \eqref{condiN} can be expressed alternatively as
	$$	{N}(n,d,t)= \{\mathbf{x}=(s_0,s_1,\dots,s_{d-1}) \in A(d) :\,\,	A\cdot \mathbf{x}^T=\mathbf{b}^T \}.$$
			Note that each line of $A$ has two 1's and others 0, and
		the spacing between the subscripts of ${s}_{(n-i) \,\text{mod}\,\,d}$ and ${s}_{(d-i)\,\text{mod}\,\,d}$ with $1\leq i\leq t+1$ is the same,
		equal to $(n-d)\,\text{mod}\,\,d$.
		Thus the $i$-th line of $A$ shall be got by taking left circular shift $i$ bits on the first line of $A$.
		So we can define the gap of $A$ as the least number of 0's between two 1's under circle in the first line of $A$, denoted by $gap(A)$.
		In fact, $gap(A)$ can completely determine the matrix $A$.
		Let $l \equiv n \,(\text{mod}\,\,d)$, then we  can  see that if $l\leq \lceil\frac{d-1}{2}\rceil$, then $gap(A)=l-1$, otherwise $gap(A)=d-1-l$.

	With the above discussions, we shall
	determine $	{N}(n,d,t)$ by the following recursive relation.
	Let $B=(A,b^T)$ denote as the augmented matrix and $rank(A)$ denote as the rank of a matrix $A$.
	In particular, for prime $d$, $rank(A)$ and $rank(B)$ can be obtained directly.
	\begin{lemma}\label{lemma se}
		%
		Let ${N}(n,d,t)$ with $t>0$ be defined in \eqref{condiN}.
		
		\noindent{\rm (i)} For composite integers $d$, let
		$\max\limits_{e|d, 1<e<d}\{e\}=\tilde{e}_d$.
		When $t < \tilde{e}_d$, we have		
		$$\begin{array}{cll}	
			|{N}(n,d,t)|=	
			\begin{cases}
				0,  &\text{ if }  d \,|\,n,  \\
				2^{d-t-1}-	\sum\limits_{e|d, 1<e<d} |{N}(n,e,t)|, &   \text{ if }  d \nmid n.     \\
			\end{cases}		
		\end{array}	$$
		When $t \geq \tilde{e}_d$, we get	
		$$\begin{array}{cll}
			|{N}(n,d,t)|=	
			\begin{cases}
				0,  &\text{ if }  d \,|\,n \text{ or if } d \nmid n \text{ and } rank(A)\neq rank(B),  \\
				2^{d-rank(A)}, &   \text{ if }  d \nmid n \text{ and } rank(A)= rank(B).     \\
			\end{cases}		
		\end{array}	$$	
		
		\noindent{\rm (ii)}	
		For prime $d$,
		it follows
		$$\begin{array}{cll}
			|{N}(n,d,t)|=
			\begin{cases}
				0, &   \text{ if }  d \,|\,n  \text{ or if } d \nmid n \text{ and }   t\geq d,\\
				2^{d-t-1},  &\text{ if } d \nmid n \text{ and } 1\leq t\leq  d-2,  \\
				2, &   \text{ if } d \nmid n \text{ and }   t=d-1.     \\
			\end{cases}		
		\end{array}	$$			
	\end{lemma}
	\begin{proof}
		(i)	If $d \,|\,n $, then the condition in the second equation of \eqref{condiN} can be rewritten as
	$\overline{s}_{n-1 \,(\text{mod}\,\,d)} =\overline{s}_{d-1}= {s}_{d-1}$, which is a contradiction.
		Thus $|{N}(n,d,t)|=	0$.

	When $t < \tilde{e}_d$, let $l = n \,(\text{mod}\,\,d)$, then $ 0\leq l<d $.
	 If $ d \nmid n $, then by \eqref{condiN} we have
		\begin{equation}\label{s_d_ccc}
		\s_{[0:d]}=	(s_0,s_1,\dots,s_{d-1}) =(s_0,\dots,s_{d-t-2})(\overline{s}_{l-t-1},{s}_{l-t},\dots,{s}_{l-2},\overline{s}_{l-1}),
		\end{equation}
	where 	$\overline{s}_{l-t-1}$ can be determined by $(s_0,\dots,s_{d-t-2})$,
	${s}_{l-t}$ can also be  determined by $(s_0,\dots,$\\
	$ s_{d-t-2},\overline{s}_{l-t-1})$.
	 This process continues, which enables us to derive that  $\s_d$ is entirely determined by $(s_0,\dots,s_{d-t-2})$.
Thus, 	among $ {2}^{d-t-1}$ sequences $(s_0,\dots,s_{d-t-2}) $, it suffices to subtract those  sequences $(s_0,\dots,s_{d-t-2})$ which do not satisfy $\s_d$ in \eqref{s_d_ccc} is aperiodic. That is,	
		$$
		\begin{array}{cll}
		&|{N}(n,d,t)|\\
	= & 2^{d-t-1}- |\{(s_0,\dots,s_{d-t-2}):  \text{the sequence in  \eqref{s_d_ccc} is a periodic finite-length sequence} \}|.
		\end{array}
	$$
		In the following, we investigate those cases which lead to the sequence  in \eqref{s_d_ccc}  is a periodic finite-length sequence, that is to say, $\s_{[0:d]}=(\s_e)^{d/e}$ where $\s_e$ is aperiodic with $e|d, 1\leq e<d$.
		When $e=1$, we have $\s_{[0:d]}=(0^d)$ or $(1^d)$, a contradiction that $\overline{s}_{(n-1) \,\text{mod}\,\,d} = {s}_{d-1}$.
		When $e>1$,
		it follows from $\s_{[0:d]}=(\s_e)^{d/e}$ that
		$$s_{(d-j)\,\text{mod}\,\,d}=s_{(d-j)\,\text{mod}\,\,e}=s_{(e-j)\,\text{mod}\,\,e}
		\text{ \, and \, } s_{(n-j)\,\text{mod}\,\,d}=s_{(n-j) \,\text{mod}\,\,e}$$
		for $1 \leq j \leq t+1$.
		Since the condition $\overline{s}_{(n-1) \,\text{mod}\,\,d} = {s}_{d-1}, {s}_{(n-i) \,\text{mod}\,\,d} = {s}_{(d-i)\,\text{mod}\,\,d}$ with $2\leq i \leq t$, and ${s}_{(n-t-1) \,\text{mod}\,\,d}= \overline{s}_{(d-t-1) \,\text{mod}\,\,d}$,
		by $\s_{[0:d]}=(\s_e)^{d/e}$ we can see that $\s_e \in A(e)$ satisfies
		$$\overline{s}_{(n-1)\,\text{mod}\,\,e} = s_{e-1},  {s}_{(n-i) \,\text{mod}\,\,e} = {s}_{(e-i)\,\text{mod}\,\,e} \text{ \, and \, }  {s}_{(n-t-1) \,\text{mod}\,\,e} = \overline{s}_{(e-t-1)\,\text{mod}\,\,e},$$
		where $A(e)$ denotes the set of all aperiodic finite-length sequences  $\s_e$
		of length $e$.
		Hence,
		\begin{equation*}
		\begin{array}{cll}
				\vspace{2mm}
			&|\{(s_0,\dots,s_{d-t-2}):  \text{the sequence in \eqref{s_d_ccc} is periodic} \}|\\
			=&\sum\limits_{e|d, 1<e<d}|\{ \s_e \in A(e): \overline{s}_{(n-1)\,\text{mod}\,\,e} = s_{e-1},
			\vspace{2mm}
			{s}_{(n-i) \,\text{mod}\,\,e} = {s}_{(e-i)\,\text{mod}\,\,e}, \,\, 2 \leq i \leq t, 	\\
				\vspace{2mm}
			& \text{ and }{s}_{(n-t-1) \,\text{mod}\,\,e}= \overline{s}_{(e-t-1) \,\text{mod}\,\,e} \}|	\\
			=& \sum\limits_{e|d, 1<e<d} |N(n,e,t)|.
		\end{array}	
	\end{equation*}	
	Therefore,
	$	|N(n,d,t)|=2^{d-t-1}-	\sum\limits_{e|d, 1<e<d} |N(n,e,t)| .$

		When $t > \tilde{e}_d$,	suppose that $\s_{[0:d]}$ is a periodic finite-length sequence, then $\s_{[0:d]}=(\s_e)^{d/e}$ where $\s_e$ is aperiodic with $e|d, 1\leq e<d$.
		Since $add(\sn)=t$, we have $s_{(d-j)\,\text{mod}\,\,d}=s_{(d-j)\,\text{mod}\,\,e}=s_{(e-j)\,\text{mod}\,\,e}$
		and $s_{(n-j)\,\text{mod}\,\,d}=s_{(n-j) \,\text{mod}\,\,e}$ for $1 \leq j \leq t+1$.
		It follows from \eqref{condiN}
		that $\overline{s}_{(n-1) \,\text{mod}\,\,e}=\overline{s}_{(n-1) \,\text{mod}\,\,d} = {s}_{d-1}= {s}_{e-1}$ and
		$${s}_{(n-1) \,\text{mod}\,\,e}={s}_{(n-e-1) \,\text{mod}\,\,e} ={s}_{(n-e-1) \,\text{mod}\,\,d} = {s}_{d-e-1}= {s}_{e-e-1}={s}_{e-1},$$
		a contradiction.
		Thus,  $\s_d$ is aperiodic.
		
		When $t = \tilde{e}_d$,
		according to the relationship between the factors of
		$d$ and $\tilde{e}_d$, we divide the discussion into two cases.
		For $e_1|d$ with $e_1<\tilde{e}_d$,
		we can get a similar contradiction as the case of $t > \tilde{e}_d$, implying that $\s_d \neq (\s_{e_1})^{d/e_1}$.
		For $e_1=\tilde{e}_d$,	since $\tilde{e}_d|d$, suppose $\s_d = (\s_{\tilde{e}_d})^{d/\tilde{e}_d}$,
		then the conditions are reduced to
		$	\overline{s}_{(n-1) \,\text{mod}\,\,\tilde{e}_d} = {s}_{\tilde{e}_d-1}$ and ${s}_{(n-i) \,\text{mod}\,\,\tilde{e}_d}=
		{s}_{(\tilde{e}_d-i)\,\text{mod}\,\,\tilde{e}_d}, 2\leq i\leq \tilde{e}_d$.
	Thus,
	$$
	\begin{array}{cll}
	& {s}_{(2n-1) \,\text{mod}\,\,\tilde{e}_d} = {s}_{(n-(\tilde{e}_d-n+1)) \,\text{mod}\,\,\tilde{e}_d}  = {s}_{(\tilde{e}_d-(\tilde{e}_d-n+1)) \,\text{mod}\,\,\tilde{e}_d}      = {s}_{(n-1) \,\text{mod}\,\,\tilde{e}_d},\\
	&  {s}_{(3n-1) \,\text{mod}\,\,\tilde{e}_d} = {s}_{(2n-(\tilde{e}_d-n+1)) \,\text{mod}\,\,\tilde{e}_d}= {s}_{(2n-1) \,\text{mod}\,\,\tilde{e}_d}, \\
	&  \dots \quad     \dots  , \\
	&  {s}_{(\tilde{e}_dn-1) \,\text{mod}\,\,\tilde{e}_d} = {s}_{((\tilde{e}_d-1)n-1) \,\text{mod}\,\,\tilde{e}_d}
	\end{array}
$$
		It implies that $${s}_{\tilde{e}_d-1}= {s}_{(\tilde{e}_dn-1) \,\text{mod}\,\,\tilde{e}_d}={s}_{((\tilde{e}_d-1)n-1) \,\text{mod}\,\,\tilde{e}_d}=\dots={s}_{(2n-1) \,\text{mod}\,\,\tilde{e}_d}={s}_{(n-1) \,\text{mod}\,\,\tilde{e}_d},$$
		which contradicts the condition $\overline{s}_{(n-1) \,\text{mod}\,\,\tilde{e}_d} = {s}_{\tilde{e}_d-1}$.
		Therefore, the assumption does not hold.
		That is, $\s_{[0:d]}$ is always aperiodic and can be denoted by $\s_d$.
		
		Therefore, when $t \geq \tilde{e}_d$ and $d \nmid n$, we have
		$$
				\begin{array}{cll}
		&|{N}(n,d,t)| \\
		=&| \{ \s_d :	\overline{s}_{(n-1) \,\text{mod}\,\,d} = {s}_{d-1}, {s}_{(n-i) \,\text{mod}\,\,d} = {s}_{(d-i)\,\text{mod}\,\,d}, {s}_{(n-t-1) \,\text{mod}\,\,d} = \overline{s}_{(d-t-1)\,\text{mod}\,\,d} \}|,
				\end{array}	
		$$
		which is equal to the number of solutions to the equation in $A\cdot \mathbf{x}^T=\mathbf{b}^T$ with $A,\mathbf{x}^T,\mathbf{b}^T$ defined above.
		Therefore, for $B=(A,\mathbf{b}^T)$, if $rank(A)= rank(B)$, then  $|{N}(n,d,t)|=2^{d-rank(A)}$; if
		$rank(A)\neq rank(B)$, then  $|{N}(n,d,t)|=0$.

		(ii)
		If 	$d$ is a prime number, when $d \,|\,n$, similar to (i), it can be shown that
		$|N(n,d,t)|=0.$
			When $d \nmid n$,
		since $A=(a_{i,j})_{(t+1)\times d}$ and $B=(A,b^T)$,  we know that
			$rank(A)\leq rank(B)\leq 1+t$.	
			Suppose $rank(A)< 1+t$.
			Then there exist $i_1,i_2,\dots,i_k$ with a minimal positive integer $2 \leq k\leq  1+t$
			such that
			\begin{equation}\label{kkkkk}
				L^{i_1}(\s)+L^{i_2}(\s)+\dots+L^{i_{k-1}}(\s)=L^{i_k}(\s),
			\end{equation}
			where 	$\s$ is the 	$d$-length vector corresponding to the first row of the matrix
			$A$.
			Thus,
			$L^{i_1}(\s)+L^{i_2}(\s)+\dots+L^{i_{k-1}}(\s)+L^{i_k}(\s)=0$.
		According to the definition of matrix $A$,
				without loss of generality,
			let
			$$L^{i_1}(\s)=(\overbrace{ 0\dots0}^{l-1}1 \,0\dots \dots 0\, \overbrace{ 0\dots0}^{l-1}1) =( 0^{l-1}1 \,0\dots \dots 0\, 0^{l-1}1 ),$$
			 where there are two 1's and others are 0's in $L^{i_1}(\s)$.
			 Hence, $gap(A)=l-1$ with $l\geq 1$.
			We denote $L^{i_1}(\s)$ by $(\delta \,0\dots\dots 0\, \delta)$ where $\delta=(0^{l-1}1)$ and others are 0's.
			Recall that $L^{i_j}(\s)$, $2 \leq j \leq k$, are circularly shifted sequences of $L^{i_1}(\s)$.
		Then since 	$L^{i_1}(\s)+L^{i_2}(\s)+\dots+L^{i_{k-1}}(\s)+L^{i_k}(\s)=0$, it follows that	
	$$\left(
	\begin{array}{c}
		L^{i_1}(\s)\\
		L^{i_2}(\s)    \\
		L^{i_{3}}(\s)\\
		\dots\\
		L^{i_k}(\s)\\
	\end{array}
	\right)
	=\left(
	\begin{array}{cll}
		&	\delta0 \cdots\cdots\cdots 0\delta \\
		&	\delta\delta0 \cdots\cdots\cdots 0   \\
		&	\cdot 	0\delta\delta0 \cdots\cdots \cdot  \cdot 0 \\
		&	\cdots	\cdots\cdots \cdots \\
		&	0\cdots\cdots\cdots	0\delta\delta\\
	\end{array}
	\right),
	$$
	where the order of $i_j, 2 \leq j \leq k$ to be chosen  appropriately. Combining this with the fact that the length of the vector 	$\s$ is 	$d$, we obtain 	$kl=d.$

			When $t  = 1,2,\dots,d-2,$
			since $kl=d$ and $k\leq 1+t \leq d-1$,
			we have $l>1$,  which leads to a contradiction with
			$d$ being prime.
			Hence, $rank(A)=1+t$,
			thus $1+t=rank(A)\leq rank(B)\leq 1+t $,
			implying
			$rank(A)=rank(B)=1+t$.
			Therefore,
			$|{N}(n,d,t)|=2^{d-rank(A)}=2^{d-t-1}$.

			When $t  = d-1,$ from $kl=d$ with $2 \leq k\leq 1+t=d$, it follows $l=1$ and $k=d$. Then
			$\s+L(\s)+L^{2}(\s)+\dots+L^{d-1}(\s)=0$ by \eqref{kkkkk}.	
			It can be seen that $rank(A)=d-1=t$.
			Since the vector $\mathbf{b}=(b_0,b_1,\dots,b_{t})=(1,\overbrace{0, \dots 0}^{t-1},1)$
			satisfies $b_0 \oplus b_1 \oplus \dots \oplus b_{t} = 0$, where  $\oplus$ represents the XOR operation, we have $rank(B)=t$.
			Thus $rank(A)=rank(B)=t$,
			implying that
			$|{N}(n,d,t)|=2^{d-rank(A)}=2$.

	When $t\geq d$,
		if $t=d$, then similarly as the proof of $t=d-1$, we obtain that $rank(A)=d-1=t-1$ and $rank(B)=t$, implying that $rank(A)\neq rank(B)$.
	Thus 	${N}(n,d,t)=0$.
	If $t> d$, then it follows from $add(\sn)=t$ that
		$$\overline{s}_{(n-1) \,\text{mod}\,\,d} = {s}_{d-1}, \,\, {s}_{(n-i) \,\text{mod}\,\,d} = {s}_{(d-i)\,\text{mod}\,\,d}, \,\,2 \leq i  \leq t,
		\,\,{s}_{(n-t-1) \,\text{mod}\,\,d} = \overline{s}_{(d-t-1)\,\text{mod}\,\,d},$$
		 we see that $\overline{s}_{(n-1) \,\text{mod}\,\,d} = {s}_{d-1}$ and
		${s}_{(n-(d+1)) \,\text{mod}\,\,d} = {s}_{(d-(d+1))\,\text{mod}\,\,d}, $
	 a contradiction. Thus $|{N}(n,d,t)|=0$.
	\end{proof}

\begin{remark}	
	From Lemma \ref{lemma se}, for prime $d$, the value of $	N(n,d,t)$ is determined by
	$rank(A)$, and $rank(A)$ is completely determined by $gap(A)$.
	When $n  \,(\text{mod}\,\,d) =l>\lceil\frac{d-1}{2}\rceil$, it follows $gap(A)=(d- l)-1 $.
	Then
	$|{N}(n,d,t)|=|{N}(n',d,t)|$ with $n' \,(\text{mod}\,\,d) =l' =d- l \leq \lceil\frac{d-1}{2}\rceil$.
	Thus, we only need to calculate $	N(n,d,t)$ with $n  \,(\text{mod}\,\,d) =l  \leq \lceil\frac{d-1}{2}\rceil$.
\end{remark}

	According to the above analysis, the value of $|{\mathcal{P}}(n,\omega)|$ can be obtained immediately below.

		\begin{theorem}\label{enum}
		For positive integers $n \geq 3$ and $\omega\geq  \NN$,
	let	${\mathcal{P}}(n,\omega)$
		be the set of binary shift inequivalent sequences with period $n$ and nonlinear complexity $\omega$.
	Then	we have
		$$\begin{array}{cll}
			|{\mathcal{P}}(n,\omega)|		
			=&\sum\limits_{d =1}^{n- \omega-1}\sum\limits_{e | d} \mu(e) 2^{ n-\omega-1-d +d/e}  +\sum\limits_{d=n- \omega}^{\n}	|N(n,d,\omega+d-n)|,\\
		\end{array}
		$$
		where $\mu(e)$ is the M${\ddot{o}}$bius function in \eqref{MoMo}, $|N(n,n- \omega,0)|$ and $|N(n,d,\omega+d-n)|$ with $d>n- \omega$ can be determined by Lemmas \ref{lemma fi} and \ref{lemma se}, respectively.
	\end{theorem}

	\begin{example}
	Take an example for $n=16$ and $\omega=12\geq \NN$,
$$\begin{array}{cll}
	& |{\mathcal{P}}(n,\omega)|\\
	=&\sum\limits_{d =1}^{n- \omega-1}\sum\limits_{e | d} \mu(e) 2^{ n-\omega-1-d +d/e} +\sum\limits_{d=n- \omega}^{\n}	|N(n,d,\omega+d-n)|\\
	=&\sum\limits_{d =1}^{3}\sum\limits_{e | d} \mu(e) 2^{ 3-d +d/e} + |N(16,4,0)|+ |N(16,5,1)|\\
&	+ |N(16,6,2)|+ |N(16,7,3)|+|N(16,8,4)|.\\
\end{array}
$$
When $d=4$, it follows from Lemma \ref{lemma fi} and $4\,|\,16$ that $|N(16,4,0)|=|A(4)|=12$.
When $d\in \{5,6,7,8\}$, according to Lemma \ref{lemma se}, for the prime number $d$ and $d\nmid n$, we have $|N(16,5,1)|=2^{5-1-1}=8$ and $|N(16,7,3)|=2^{7-3-1}=8$; for the composite number $d$, we see
$$|N(16,6,2)|=2^{6-2-1}-|N(16,2,2)|-|N(16,3,2)|=8-0-2=6, $$
and $|N(16,8,4)|=0$ by $8\,|\,16$.
Thus it yields $|{\mathcal{P}}(n,\omega)|	=18+12+8+6+8+0=52$.
The result is consistent with the exhaustive
search presented in \cite[Table~3.2]{JansenPhD}.
\end{example}

\begin{remark}
From Theorem \ref{enum} of this paper, we know that by using the formulas in
Lemmas \ref{lemma fi} and \ref{lemma se}, the exact value of
$|{\mathcal{P}}(n,\omega)|	$, $\omega\geq \NN$, can be calculated directly.
Sun et al. \cite{SZLH} and Xiao et al. \cite{Xiao2018} presented $|{\mathcal{P}}(n,n-1)|	$ and $|{\mathcal{P}}(n,n-2)|$,
which also can be calculated by our method.
 Furthermore,  the exact probability that a random $n$-periodic binary
sequence achieves the nonlinear complexity $\omega\geq  \NN$ can be determined by Theorem \ref{enum} as
 $$\frac{n|{\mathcal{P}}(n,\omega)|	}{|A(n)|}=\frac{n|{\mathcal{P}}(n,\omega)|}{\sum\limits_{e | n} \mu(e) 2^{n/e}}$$
Erdmann and Murphy
in \cite{Erdmann} derived an approximate probability distribution of the
nonlinear complexity, the accuracy of the approximate probability can be evaluated by
comparison with our result.

It is to be noted that the method proposed in this paper does not apply to the case $\nup \leq \omega< \NN$. The main reason is that some finite-length sequences used to generate periodic sequences may not be representative, making it intractable to identify the representative sequences. Furthermore, some representative sequences are shift-equivalent and thus not unique. These factors together hinder an accurate characterization of the structure and enumeration of the corresponding periodic sequences.
\end{remark}

\section{Conclusion}\label{Sec6}	
In this paper,
 we proved that
	  the finite-length sequences that generate shift inequivalent $n$-periodic sequences with nonlinear complexity larger than or equal to $ \NN $ are the unique  representative sequences.
	  This enabled us to determine both the structure  and the exact enumeration of such periodic sequences.
In the future work,	more randomness properties of those periodic sequences,
	such as Golomb's three randomness postulates (balance, run, autocorrelation properties), $k$-error nonlinear complexity,
	correlation measure of order $k$, will be  investigated.

\appendix
\section*{Appendix A:  Proof of Lemma \ref{b range} }\label{appendix:lemmab}	
\begin{proof}
	Let $a=n-c-d_1$, then $h=a+b$.
	Thus $\textbf{v}_n=R^{h}(\textbf{s}_{n})=R^{a+b}(\textbf{s}_{n})$.
  From the first paragraph proof of Lemma 8 in \cite{yuan}, 	it follows $0< b\leq c+d_{2}-1$ when $n$ is even.
  When $c= \n$ with odd $n$, we claim that $0\leq b <c+d_2$.

  Let $n=2k+1$, then $c=k$.
  It follows from $\textbf{s}_{2k+1}\in \mathcal{B}(2k+1, k, d_1)$ that	
  $(s_0,s_1,\dots,s_{k-2})$ and $(s_{d_1},s_{d_1+1},\dots,s_{d_1+k-2})$ is a pair of  identical subsequences of length $k -1$ with different successors.	
  Let $\textbf{v}_{n}=R^{h}(\textbf{s}_{n})$, then
  suppose that $h\leq n-c-d_1-1=k-d_1$.
  Due to $\textbf{v}_{2k+1}\in \mathcal{B}(2k+1,k, d_2)$ with $d_2\leq k$,
  we have $\textbf{v}_{2k}\in \mathcal{B}(2k,k, d_2)$.
  Thus  $\textbf{v}_{2k}\in \mathcal{B}(2k,k, d_2)$ has
  a pair of identical sequences
  $(v_0,v_1,\dots,v_{k-2})$ and $(v_{d_2},v_{d_2+1},\dots,v_{d_2+k-2})$
  of length $k -1$ with different successors
  in the beginning, and has another pair
  $$
  \left\{\begin{array}{cll}
  	(s_0,s_1,\dots,s_{k-2})   & =(v_h,v_{h+1},\dots,v_{h+k-2}), \\
  	(s_{d_1},s_{d_1+1},\dots,s_{d_1+k-2}) & =(v_{h+d_1},v_{h+d_1+1},\dots,v_{h+d_1+k-2}).
  \end{array}\right.
  $$
  with different successors behind,
  where  $h+d_1+k-2\leq (k-d_1)+d_1+k-2=2k-2$.
  It contradicts the fact in Lemma \ref{lem unique}. Thus we have $h\geq n-c-d_1$.
  For $\textbf{v}_{n}=R^{h}(\textbf{s}_{n})\in \mathcal{B}(n,c,d_2)$ and $R^{h'}(\textbf{v}_{n})\in \mathcal{B}(n,c,d_1)$, we see $h'=n-h \geq n-c-d_2$, i.e., $h \leq c+d_2$.
  Hence we obtain $n-c-d_1 \leq h=(n-c-d_1)+b \leq c+d_2$.
  From $n-c-d_1\geq 1$, it
  implies $0\leq b< c+d_{2}$. Hence, the claimed statement follows.

  Thus,  we have	 $0< b\leq c+d_{2}-1$ when $n$ is even, and $0\leq b\leq c+d_{2}-1$  when $n$ is odd.
	Then it suffices to investigate $\textbf{v}_n=R^{a+b}(\textbf{s}_{n})$ with $0\leq b\leq c+d_{2}-1$,
	where when $b=0$, $c=\n$ with odd $n$.
	Note that $	\textbf{v}_{n}=R^{a+b}(\textbf{s}_{n})$  has the following form,
	\begin{equation}\label{shizi-1}
		\begin{array}{cll}
			\textbf{v}_{n}
			&=(\,\overbrace{({v_{0}, \ \dots, \ v_{c+d_{2}-1}\,})}^{c+d_{2}}
			\ ({v_{c+d_{2}},  \ \dots, \ v_{n-t_{2}-1}})\, \overbrace{({v_{n-t_{2}}, \ \dots\,, \   v_{n-1}}\,)}^{t_{2}}\,)\\
			=R^{a+b}(\textbf{s}_{n})&=(\,\underbrace{ (v_{0}, \ \dots, \ v_{b-2}, \ v_{b-1})}_{b}\, (v_{b}, \ \dots, \ v_{b+a-t_{1}-1}\,)
			\,\underbrace{(v_{b+a-t_{1}},\ \dots, \ v_{n-1}\,)}_{t_{1}+c+d_{1}-b}\,).
		\end{array}
	\end{equation}
	Since
	$\textbf{v}_{n}\in \mathcal{B}(n,c,d_{2})$  with $add(\textbf{v}_{n})=t_{2}$, one can get
	\begin{equation}\label{i_2}
		v_{i}={v}_{i+d_{2}}    \,\, \mbox{for}  \,\, i \in [-t_{2}, c-2]
		\,\, \mbox{and} \,\, v_{c-1}=\overline{v}_{c+d_{2}-1}.
	\end{equation}
	Moreover,
	since $\textbf{s}_{n}\in\mathcal{B}(n,c,d_{1})$ with $add(\textbf{s}_{n})=t_{1}$,
	one has $$ s_{i}={s}_{i+d_{1}}    \,\, \mbox{for}  \,\, i \in [-t_{1}, c-2]
	\,\, \mbox{and} \,\,
	s_{c-1}=\overline{s}_{c+d_{1}-1}.$$
	Due to ${v}_{i+a+b}={s}_{i}$, we have
	\begin{equation}\label{i_1}
		v_{i}={v}_{i+d_{1}}    \,\, \mbox{for}  \,\, i \in [b+a-t_{1}, b+n-d_{1}-2]
		\,\, \mbox{and} \,\,
		v_{b+n-d_{1}-1}=\overline{v}_{b+n-1},
	\end{equation}
	where the subscripts in above equations are taken modulo $n$.	
	Thus a subsequence consisting of arbitrary consecutive $d_1$ terms in $(v_{b+a-t_{1}},\dots,v_{n-1},v_{n},\dots,v_{b+n-2},\overline{v}_{b+n-1})$ is a shifted version of $\textbf{s}_{d_{1}}$.
	According to the form of $\textbf{v}_n=R^{a+b}(\textbf{s}_{n})$ in \eqref{shizi-1}, we shall utilize (\ref{i_2}) and (\ref{i_1}) to show that  the subsequence $\textbf{s}_{d_{1}}$ of $\textbf{s}_{n}$ has two representations with different Hamming weights, which is a clear contradiction.
	In what follows, we divide the discussion into two cases for different values of $b$ with $ 0\leq b\leq c+d_{2}-1$.

	\begin{figure}[!h]
		\centering
		\begin{tikzpicture}	
			\node (rect) at (0,0) {
				$\begin{array}{cll}
					\textbf{v}_{n}
					& = (\,\overbrace{\, v_{0},\, v_{1}, \, \ldots, \,   v_i, \,  v_{i+1},  \,  \ldots,\,  v_{c+d_{2}-2},\, v_{c+d_{2}-1} \,}^{c+d_{2}},
					\ldots,     v_{n - t_{2}},\, \ldots,  v_{n-1}) \\	
					= R^{a+b}(\textbf{s}_{n})
					& = (\,\underbrace{\, v_{0}, \, \ldots,\, v_{b-1}, \, v_{b},\ \ldots, \, \, v_{b+a - t_{1}}, \,  \ldots,\ v_{b+a - t_{1} + d_1 + d_2 - 1} \,}_{b + (a - t_{1}) + d_{1} + d_{2}}, \ldots, \ v_{n-1} ).
				\end{array}$
			};
			
			\draw[<->, thick] (1.2,-1.3) -- (3.2,-1.3) node[above] at (2,-2) {$d_1$-length};				
			\draw[<-, very thick] (0,-1.5) -- (0.8,-1.5) node[above] at (0.3,-2.2) {$-d_{2}$};
			
			\draw[dashed] (3.2,-1) -- (3.2,0);
			\draw[dashed] (1.2,-1) -- (1.2,0);
		\end{tikzpicture}
		\caption{The visualized description of Case (1) }\label{fig.0}
	\end{figure}

	{\bf Case (1)}:
	If $0\leq b< t_{1}$,
	then $b+a-t_{1}\leq c-d_{1}$.
	Consider the $d_{1}$-tuple subsequence $\textbf{v}_{[c+d_{2}-d_{1}:c+d_{2}]}$, from
	(\ref{i_2})
	we have
	$$\textbf{v}_{[c+d_{2}-d_{1}:c+d_{2}]}=(v_{c+d_{2}-d_{1}},\dots,v_{c+d_{2}-2},v_{c+d_{2}-1})=(v_{c-d_{1}},\dots,v_{c-2},\overline{v}_{c-1}).$$
	Recall that arbitrary consecutive $d_1$ terms in $(v_{b+a-t_{1}},\dots,v_{n-1},v_{n},\dots,v_{b+n-2},\overline{v}_{b+n-1})$ compose of a shifted version of $\textbf{s}_{d_{1}}$.
	Since $b+a-t_{1}\leq c-d_{1}$ and $c+d_2 \leq n+b-1$, it follows $\textbf{v}_{[c-d_{1}:c]}$ and $\textbf{v}_{[c+d_{2}-d_{1}:c+d_{2}]}$ are shifted versions of $\textbf{s}_{d_{1}}$.
	Combining these, their Hamming weights satisfy
	\begin{equation*}\label{wt-1}
		wt(\textbf{s}_{d_{1}})=wt(\textbf{v}_{[c-d_{1}:c]})=wt(\textbf{v}_{[c+d_{2}-d_{1}:c+d_{2}]})=wt((v_{c-d_{1}},\dots,v_{c-2},\overline{v}_{c-1})),
	\end{equation*}
	a contradiction.

	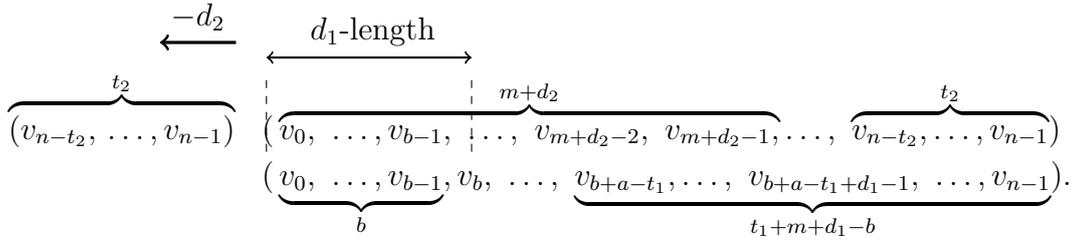
\begin{figure}[!h]
		\centering
		\begin{tikzpicture}	
			\node (rect) at (0,0)
			{$\begin{array}{cll}	
					\overbrace{(v_{n-t_{2}}, \, \dots, v_{n-1})}^{t_{2}}&(\,\overbrace{{v_{0}, \ \dots, v_{b-1}, \ \dots, \ v_{m+d_{2}-2}}, \ v_{m+d_{2}-1},}^{m+d_{2}}
					\dots, \ \overbrace{v_{n-t_{2}},  \dots,  v_{n-1}}^{t_{2}})\\	
					&(\,\underbrace{v_{0}, \   \dots,  v_{b-1} }_{b},v_{b}, \ \dots, \ \underbrace{v_{b+a-t_{1}}, \dots,\ v_{b+a-t_{1}+d_1-1}, \
						\dots,   v_{n-1}}_{t_{1}+m+d_{1}-b}).
				\end{array} $};
			\draw[<->,thick](-3.6,1.3)--(-0.9,1.3) node[above] at (-2.2,1.3)  {$d_1$-length};
			\draw[<-, very thick] (-5,1.5)--(-4,1.5) node[above] at (-4.5,1.5) {$-d_{2}$};
			\draw[dashed] (-3.6,1)--(-3.6,0);
			\draw[dashed] (-0.9,1)--(-0.9,0);			
		\end{tikzpicture}
		\caption{The visualized description of Case (2) }\label{fig.1}
	\end{figure}

	{\bf Case (2)}:
	If $d_{1}+d_{2} \leq b+t_{2}$ and $b< c+d_{2}$, then $d_{1}+d_{2}-t_{2}\leq b <c+d_{2}$. Consider the $d_{1}$-tuple subsequence $\textbf{v}_{[b-d_{1}-d_{2}:b-d_{2}]}$, since $-t_2 \leq b-d_{1}- d_2$ and $ b -d_{2}<c$, from (\ref{i_2}) we have
	$$\textbf{v}_{[b-d_{1}-d_{2}:b-d_{2}]}=(v_{b-d_{1}-d_{2}},\dots,v_{b-d_{2}-2},v_{b-d_{2}-1})=(v_{b-d_{1}},\dots,v_{b-2},{v}_{b-1})=\textbf{v}_{[b-d_{1}:b]}.$$
	Recall that arbitrary consecutive $d_1$ terms in $(v_{b+a-t_{1}-n},\dots,v_{-1},v_{0},\dots,v_{b-2},\overline{v}_{b-1})$ compose of a shifted version of $\textbf{s}_{d_{1}}$.
	Due to $b-d_{1}-d_{2} \geq b+a-t_{1}-n$ and $b-d_{2} < b$,
	we obtain that
	$\textbf{v}_{[b-d_{1}-d_{2}:b-d_{2}]}$
	and $(v_{b-d_{1}},\dots,v_{b-2},\overline{v}_{b-1})$
	are shifted versions of $\textbf{s}_{d_{1}}$.
	Combine the above equality, we have
	\begin{equation*}
		wt(\textbf{s}_{d_{1}})=wt(v_{b-d_{1}},\dots,v_{b-2},\overline{v}_{b-1})=wt(\textbf{v}_{[b-d_{1}-d_{2}:b-d_{2}]})=wt(\textbf{v}_{[b-d_{1}:b]}),
	\end{equation*}
	which is a contradiction.

	Therefore, it follows that only $ t_{1} \leq b< d_{1}+d_{2}-t_{2}$ can be possible.	
\end{proof}

\section*{Appendix B:  Proof of Lemma \ref{d_2=b}}
\begin{proof}
	From Lemma \ref{b range}, for	$R^b(\sn)\in \mathcal{B}(n,c,d_2)$, it suffices to consider $ t_{1} \leq b< d_{1}+d_{2}-t_{2}$.  	
	Let $R^b(\sn)=\textbf{v}_n$. Then $v_{i+b}=s_i$, where the subscripts are taken modulo $n$.
	Suppose $b \neq d_2$, then we will derive contradictions that two shifted versions of a same subsequence $\textbf{v}_{d_2}$ (or $\textbf{s}_{d_1}$) have different Hamming weights.
	Since
	$\textbf{v}_{n}\in \mathcal{B}(n,c,d_{2})$  with $add(\textbf{v}_{n})=t_{2}$, one can get
	\begin{equation}\label{fd_2}
		v_{-t_{2}-1}=\overline{v}_{d_{2}-t_{2}-1},	\,\,  	v_{i}={v}_{i+d_{2}}    \,\, \mbox{for}  \,\, i \in [-t_{2}, c-2]
		\,\, \mbox{and} \,\, v_{c-1}=\overline{v}_{c+d_{2}-1}.
	\end{equation}
	Moreover,
	since $\textbf{s}_{n}\in\mathcal{B}(n,c,d_{1})$ with $add(\textbf{s}_{n})=t_{1}$,
	one has
	$$ s_{-t_{1}-1}=\overline{s}_{d_{1}-t_{1}-1},	\,\,  s_{i}={s}_{i+d_{1}}    \,\, \mbox{for}  \,\, i \in [-t_{1}, c-2]
	\,\, \mbox{and} \,\,
	s_{c-1}=\overline{s}_{c+d_{1}-1}.$$
	From $s_i=v_{i+b}$, we have
	\begin{equation}\label{d_1}
		v_{b-t_{1}-1}=\overline{v}_{b+d_{1}-t_{1}-1},	\,\,  	v_{i}={v}_{i+d_{1}}    \,\, \mbox{for}  \,\, i \in [b-t_{1}, b+c-2]
		\,\, \mbox{and} \,\,
		v_{b+c-1}=\overline{v}_{b+c+d_1-1}.
	\end{equation}

	From  \eqref{fd_2},
	a subsequence consisting of arbitrary consecutive $d_{2}$ terms in  $(\overline{v}_{-t_{2}-1},{v}_{-t_2},$\\
	$v_{-t_2+1},\dots,{v}_{c+d_{2}-2},\overline{v}_{c+d_{2}-1})$
	is a shifted version of $\textbf{v}_{d_{2}}=(v_0,v_{1},\dots,{v}_{d_{2}-1})$.
	Similarly, by \eqref{d_1},
	a subsequence consisting of	arbitrary consecutive $d_{1}$ terms in  $(\overline{v}_{b-t_{1}-1},{v}_{b-t_1},v_{b-t_1+1},
	\dots,$\\
	${v}_{b+c+d_{1}-2},\overline{v}_{b+c+d_{1}-1})$
	is a shifted version of $\textbf{s}_{d_{1}}=(s_0,s_{1},\dots,{s}_{d_{1}-1})$.
	In the following, we shall investigate different representations of shifted sequences of $\textbf{v}_{d_{2}}$ and $\textbf{s}_{d_{1}}$.
	Note that $d_2 \leq d_1 \leq c$ and $t_2 \geq t_1 \geq 0$. According to the
	value of $b$, we divide the discussion into four cases.

	{\bf Case (1)}: $d_2+1\leq b\leq c+t_1$.
	It is easily checked that $(v_c,v_{c+1},\dots,\overline{v}_{c+d_{2}-1})$ is a shifted version of $\textbf{v}_{d_{2}}$.
	Since  $d_2+1\leq b\leq c+t_1$, the above subscripts satisfy $b-t_1\leq c, c+1, \dots, c+d_2-1 \leq b+c-2$.
	From \eqref{d_1}, it follows from $c+d_1=n$ that
	$$ (v_c,v_{c+1},\dots,\overline{v}_{c+d_{2}-1})= (v_{c+d_1},v_{c+d_1+1},\dots,\overline{v}_{c+d_1+d_{2}-1}) =(v_{0},v_{1},\dots,\overline{v}_{d_{2}-1}).$$
	Thus $ (v_{0},v_{1},\dots,\overline{v}_{d_{2}-1})$ is a shifted sequence of $\textbf{v}_{d_{2}}= (v_{0},v_{1},\dots,{v}_{d_{2}-1})$,
	a contradiction.

	{\bf Case (2)}:  $c+t_1+1\leq b< d_1 + d_2 -t_2$.
	Since $d_1+d_2 < c+t_2+t_1$, the case is impossible.

	{\bf Case (3)}:
	$d_2-t_2\leq b<d_2$.
	Since $-t_2\leq c+b-d_2,c+b-d_2+1,\dots, c+b-1\leq c+d_2-2$, it is clear that
	$({v}_{c+b-d_2},v_{c+b-d_2+1},\dots,{v}_{c+b-1})$ is a shifted sequence of $\textbf{v}_{d_{2}}$ from \eqref{fd_2}.
	Based on $d_2 \leq c$ and \eqref{d_1}, the subscripts satisfy
	$b-t_1\leq c+b-d_2, c+b-d_2+1, \dots, c+b-2 \leq c+b-2$,
	and ${v}_{c+b-1}=\overline{v}_{c+b-1+d_1}$.
	Then we have
	\begin{equation*}
		\begin{array}{cll}
			& ({v}_{c+b-d_2},v_{c+b-d_2+1},\dots,{v}_{c+b-2},{v}_{c+b-1})\\
			=&( {v}_{c+b+d_1-d_2},v_{c+b+d_1-d_2+1},\dots,{v}_{c+b+d_1-2},\overline{v}_{c+b+d_1-1})\\
			=& ( v_{b-d_2},v_{b-d_2+1},\dots,{v}_{b-2},\overline{v}_{b-1}),\\
		\end{array}
	\end{equation*}
	which implies $(v_{b-d_2},v_{b-d_2+1},\dots,{v}_{b-2},\overline{v}_{b-1})$ is a shifted sequence of $\textbf{v}_{d_{2}}$.
	As a matter of fact, based on $-t_2\leq b-d_2, b-d_2+1, \dots, b-1 \leq c+d_2-2$ and \eqref{fd_2},
	we have
	$\textbf{v}_{[b-d_2:d]}=(v_{b-d_2},v_{b-d_2+1},\dots,{v}_{b-2},{v}_{b-1})$ is a shifted sequence of $\textbf{v}_{d_{2}}$, a contradiction.

	{\bf Case (4)}:
	$t_1 \leq b\leq d_2-t_2-1$.
	Consider the sequence $ (v_{c+d_2-d_1},v_{c+d_2-d_1+1}, \dots, v_{c+d_2-1} )$.
	According to \eqref{d_1}, since $ b-t_1 \leq c+d_2-d_1,c+d_2-d_1+1, \dots, c+d_2-1 \leq b+n-2$,
	we can see that $ (v_{c+d_2-d_1},v_{c+d_2-d_1+1}, \dots, v_{c+d_2-1}) $ is a shifted sequence of $\textbf{s}_{d_{1}}$.
	From \eqref{fd_2}, it follows that
	$$ (v_{c+d_2-d_1},v_{c+d_2-d_1+1}, \dots, v_{c+d_2-1}) = ( v_{c-d_1},v_{c-d_1+1}, \dots, v_{c-2}, \overline{v}_{c-1}).$$
	Thus, $( v_{c-d_1},v_{c-d_1+1}, \dots, v_{c-2}, \overline{v}_{c-1})$ is a shifted sequence of $\textbf{s}_{d_{1}}$.
	Moreover, due to $d_1+d_2 < c+t_2+t_1$, we have $ b-t_1 \leq c-d_1$.
	Thus from \eqref{d_1}, it is clear that $( v_{c-d_1},v_{c-d_1+1}, \dots, v_{c-2}, \overline{v}_{c-1})$
	is a shifted sequence of $\textbf{s}_{d_{1}}$, 	a contradiction.

	Finally,  	combining the four cases above and applying the condition $ t_{1} \leq b< d_{1}+d_{2}-t_{2}$, we conclude that
	$b=d_2$.
\end{proof}

\section*{Appendix C:  Proof of the tight bounds in Theorem \ref{thm_lowerbound}} \label{Appendix-lemm}	
\noindent\textbf{Proof of the tight bounds in the cases of (iii) $n=8k+4$ and (iv) $n=8k$:}

			{\bf	(iii)}	When $n=8k+4$,
let $c \geq c_0=6k+1$ and $m=\n=4k+2$.
Suppose that there is  a pair of shift equivalent sequences $(\textbf{s}'_n,\textbf{v}'_n)$ satisfying
$0 \leq add(\textbf{s}'_n)<add(\textbf{v}'_n)$, then  we have
$\sn=L^{c-m}(\textbf{s}'_n)\in \mathcal{B}(n, m,d_1)$ and $\textbf{v}_n=L^{c-m}(\textbf{v}'_n)\in \mathcal{B}(n, m,d_2)$ from Lemma \ref{c-c+1}.
Moreover, by \eqref{threecon}, some parameters satisfy
\begin{equation}\label{range_para}
	\begin{cases}
		& d_1 +d_2 < 2(n-c) \leq 4k+6, \ \ \	k \geq 2 \\
		& add(\sn) =t_1\geq c-m\geq 2k-1,  \ \ \ add(\textbf{v}_n)= t_2 \geq t_1+1\geq 2k.
	\end{cases}
\end{equation}
Let $\textbf{v}_n=R^{a+b}(\sn)$, $a=n-m-d_1$.
Thus from Lemma \ref{b range}, it suffices to investigate the case of $t_{1}<b<d_{1}+d_{2}-t_{2}$, that is $2k \leq b \leq 2k+4$.

Due to
$b<d_{1}+d_{2}-t_{2}$, we have $n-t_2-1 > b+a-t_1$.
And it follows from $d_1 \leq t_1+t_2$ that $b<d_{1}+d_{2}-t_{2}\leq d_2+t_1$, implying $b+(a-t_1)+d_1<m+d_2$.
Hence $\textbf{v}_{n}=R^{a+b}(\textbf{s}_{n})$ has the following form:
\begin{equation}\label{case_3_v_n}
	\begin{array}{cll}
		\textbf{v}_{n}
		&=(\,\overbrace{{v_{0}, \, v_{1}, \  \dots,   v_{i}, \,   v_{i+1},   \dots, \, v_{m+d_{2}-2}}, \, v_{m+d_{2}-1},}^{m+d_{2}}
		\dots, \overbrace{{v_{n-t_{2}},\dots,v_{n-1}}}^{t_{2}}\,)\\
		=R^{a+b}(\textbf{s}_{n})&=(\,\underbrace{v_{0}, \ \dots, \ v_{b-1} }_{b},  \ v_{b}, \ \dots,\, \   \underbrace{v_{b+a-t_{1}}, \ \dots,\ v_{b+a-t_{1}+d_1-1}
			,\ \dots,\ v_{n-1}}_{t_{1}+m+d_{1}-b}\,).
	\end{array}
\end{equation}
Recall that we have
\begin{equation}\label{i_2-fulv}
	v_{-t_2-1}=\overline{v}_{d_{2}-t_2-1},    \,\,	v_{i}={v}_{i+d_{2}}    \,\, \mbox{for}  \,\, i \in [-t_{2}, m-2]
	\,\, \mbox{and} \,\, v_{m-1}=\overline{v}_{m+d_{2}-1}.
\end{equation}
and
\begin{equation}\label{i_1-fulv}
	v_{b+a-t_{1}-1}=\overline{v}_{d_{1}+b+a-t_{1}-1},    \,\,	v_{i}={v}_{i+d_{1}}    \,\, \mbox{for}  \,\, i \in [b+a-t_{1}, b+n-d_{1}-2]
	\,\, \mbox{and} \,\,
	v_{b+n-d_{1}-1}=\overline{v}_{b+n-1}.
\end{equation}
According to \eqref{case_3_v_n} and \eqref{i_1-fulv},
each sequence in the sequences groups $\mathcal{A}$ and $\mathcal{B}$ is a shift equivalent sequences of $\textbf{s}_{d_{1}}$.
\begin{equation}\label{system_A}
	\qquad	\qquad	\qquad	\mathcal{A}=\left\{\begin{array}{cll}
		&(\overline{v}_{b+a-t_1-1},v_{b+a-t_1},\dots,{v}_{b+a-t_1+d_{1}-2}),\\
		&(v_{b+a-t_1},\dots,{v}_{b+a-t_1+d_{1}-1}),\\
		&\quad  \quad \dots \quad  \quad  \dots\\
		&(v_{m+d_{2}-d_{1}},\dots,{v}_{m+d_{2}-1}), \\	
	\end{array}\right.  \ \
\end{equation}
\begin{equation}\label{system_B}
	\mathcal{B}=\left\{\begin{array}{cll}
		&(v_{-t_{2}-1},\dots,{v}_{-t_{2}+d_{1}-2}), \\
		&(v_{-t_{2}},\dots,{v}_{-t_{2}+d_{1}-1}), \\
		&\quad  \quad \dots \quad  \quad  \dots\\
		&(v_{b-d_{1}-1},\dots,{v}_{b-2}), \\
		&(v_{b-d_1},\dots,{v}_{b-2},\overline{v}_{b-1}),\\
	\end{array}\right.  \ \
\end{equation}	
Note that there will always exist a positive integer $i$ such that $m-t_{1}\leq d_{2}i< m+d_{2}-t_{1}$.
Below,		according to
the relationship of sequences between the sequences groups $\mathcal{A}$ and $\mathcal{B}$, we
will give  two different representations of $\textbf{s}_{d_{1}}$, and then prove that
two representations have different Hamming weights, which leads to a contradiction.
We divide the discussion into three cases according to the relation
of $d_{2}i$ and $m+d_{2}-b$.

{\bf Case (1)}: If $m-t_{1}\leq d_{2}i< m+d_{2}-b$,
then consider
$(v_{b+d_{2}i-d_1},\dots,v_{b+d_{2}i-1})$, where $b+a-t_1\leq b+d_{2}i-d_1 < m+d_2-d_1$.
Thus from the sequences group $\mathcal{A}$ in \eqref{system_A}, $(v_{b+d_{2}i-d_1},\dots,v_{b+d_{2}i-1})$ is a shifted version of $\textbf{s}_{d_{1}}$.
It follows from $d_1 \leq t_1+t_2$ and $b>t_1$ that $b-d_1 \geq -t_2$.
Moreover $-t_2<b+a-t_1\leq b+d_{2}i-d_1 < m+d_2-d_1$.
Hence the relation in \eqref{i_2-fulv} implies that
$$(v_{b+d_{2}i-d_1},\dots,v_{b+d_{2}i-1})=(v_{b-d_1},\dots,v_{b-1})$$
is also a shift equivalent sequence of
$\textbf{s}_{d_{1}}$,
which
contradicts that
$(v_{b-d_1},\dots,{v}_{b-2},\overline{v}_{b-1})$ is a shifted version of $\textbf{s}_{d_{1}}$,
from the sequences group $\mathcal{B}$ in \eqref{system_B}.

{\bf Case (2)}: If $m+d_{2}-b< d_{2}i< m+d_{2}-t_1$,
then consider the subsequence $(v_{m+d_2-d_{2}i-d_1},\dots,v_{m+d_2-d_{2}i-1})$.
From $d_1 \leq t_1+t_2$, we have
$$-t_2 \leq t_1-d_1  < m+d_2-d_{2}i-d_1 < b-d_1 < m+d_2-d_1.$$
Thus from the sequences group $\mathcal{B}$ in \eqref{system_B},  $(v_{m+d_2-d_{2}i-d_1},\dots,v_{m+d_2-d_{2}i-1})$ is a shift equivalent sequence of $\textbf{s}_{d_{1}}$.
Recall that  $-t_2  < m+d_2-d_{2}i-d_1 < m+d_2-d_1$,
then it follows from \eqref{i_2-fulv} that
$$(v_{m+d_2-d_{2}i-d_1},\dots,v_{m+d_2-d_{2}i-1})=(v_{m+d_2-d_1},\dots,\overline{v}_{m+d_2-1}),$$
is also a shifted version of $\textbf{s}_{d_{1}}$,
which
contradicts that $(v_{m+d_2-d_1},\dots,v_{m+d_2-1})$ is a shifted version of $\textbf{s}_{d_{1}}$
by the sequences group $\mathcal{A}$ in \eqref{system_A}.

{\bf Case (3)}: If $d_{2}i= m+d_{2}-b$, then let $(-t_{2}-1)+(m+d_{2}-b)=g$.
It follows from $t_{1}<b<d_{1}+d_{2}-t_{2}$ that $g$ and $b+a-t_{1}-1$ belong to the same interval $(m-d_1-1, m+d_2-t_1-t_2-1)$. Since we can not determine which one is bigger, there are three possible cases needed to consider. 

\begin{figure}[!h]
	\centering
	\begin{tikzpicture}	
		\node (rect) at (0,0)
		{
			$				\begin{array}{cll}
				
				\overbrace{(v_{n-t_{2}-1},\dots,v_{n-1})}^{t_{2}+1}&(\,\overbrace{{v_{0},  v_{1}, \dots, v_i,   \dots, v_{m+d_{2}-2}}, v_{m+d_{2}-1},}^{m+d_{2}}			\dots, \overbrace{v_{n-t_{2}-1},  \dots, v_{n-1}}^{t_{2}+1})\\	
				&(\,\underbrace{v_{0},\dots, v_{b-1} }_{b},v_{b},\dots,\underbrace{  v_{b+a-t_{1}},  \dots,v_{b+a-t_{1}+d_1-1},
					\dots,  v_{n-1}}_{t_{1}+m+d_{1}-b}).
			\end{array} $
		};
		\draw[->, very thick] (-3.5,1.5)--(-1.2,1.5) node[above] at (-2.5,1.5) {$+d_{2}i$};
		\draw[<->,thick](-7,1.3)--(-3.5,1.3) node[above] at (-5,1.3) {$d_1$-length};
		\draw[dashed] (-7,1)--(-7,0);
		\draw[dashed] (-3.5,1)--(-3.5,0);			
	\end{tikzpicture}
	\caption{The visualized description of Subcase (3.1) }\label{fig.2}
\end{figure}
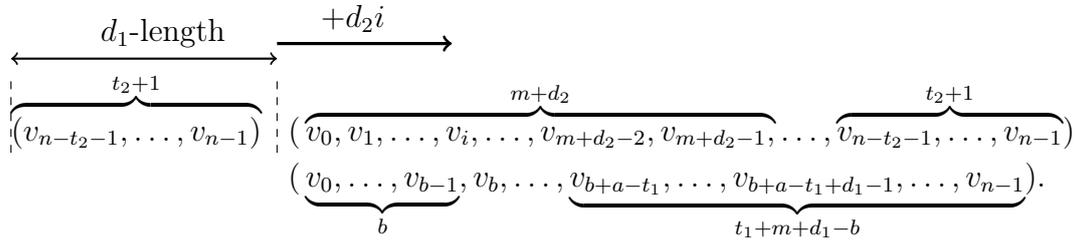
{\bf Subcase (3.1)}:	$b+a-t_{1}-1  < g.$ Since $b>t_1>d_1-t_2-1$ by $d_1 \leq t_1+t_2$, one can get that $g<m+d_2-d_1$.
Thus $b+a-t_{1}-1  < g<m+d_2-d_1$, implying that
$({v}_{g},\dots,{v}_{g+d_{1}-1})$ is a shifted version of $\s_{d_1}$ according to $\mathcal{A}$ in \eqref{system_A}.
For $(v_{-t_{2}-1},\dots,{v}_{-t_{2}+d_{1}-2})$ as a shift equivalent sequence of $\textbf{s}_{d_{1}}$, since $g=(-t_{2}-1)+d_{2}i <m+d_2-d_1$, by \eqref{i_2-fulv} we have
$$(v_{-t_{2}-1},\dots,{v}_{-t_{2}+d_{1}-2})=(\overline{v}_{g},\dots,{v}_{g+d_{1}-1}),$$
which contradicts that $({v}_{g},\dots,{v}_{g+d_{1}-1})$ is a shifted version of $\s_{d_1}$.

\begin{figure}[!h]
	\centering
	\begin{tikzpicture}	
		\node (rect) at (0,0)
		{
			$				\begin{array}{cll}
				\textbf{v}_{n}&=(\overbrace{{v_{0},  v_{1}, \dots,   v_i,   v_{i+1},  \dots,  v_{m+d_{2}-2}},  v_{m+d_{2}-1},}^{m+d_{2}}
				\dots, \overbrace{v_{n-t_{2}}, \dots, v_{n-1}}^{t_{2}})\\	
				=R^{a+b}(\textbf{s}_{n})&=	(\underbrace{v_{0}, \dots,  v_{b-1} }_{b}, v_{b},  \dots,  \underbrace{v_{b+a-t_{1}},  \dots, v_{b+a-t_{1}+d_1-1}, 				\dots, v_{n-1}}_{t_{1}+m+d_{1}-b}\,).
			\end{array} $
		};
		\draw[->, very thick] (3,-1.5) --(5,-1.5)node[above] at (4,-2.2) {$+d_{1}i$}node[above] at (4,-2.7) {subcase (3.3)};
		
		\draw[<-, very thick] (-2,-1.5)--(0,-1.5) node[above] at (-1,-2.2) {$-d_{2}i$}node[above] at (-0.8,-2.7) {subcase (3.2)};
		\draw[<->,thick](0,-1.3)--(3,-1.3) node[above] at (1.5,-2) {$d_1$-length};
		\draw[dashed] (0,-1)--(0,0);
		\draw[dashed] (3,-1)--(3,0);	
		
	\end{tikzpicture}
	\caption{The visualized description of Subcase (3.2) and Subcase (3.3) }\label{fig.3}
\end{figure}

{\bf Subcase (3.2)}: $b+a-t_{1}-1 > g.$ Then $2b+t_2> d_1+d_2+t_1$.
Consider the subsequence $(\overline{v}_{b+a-t_1-1},v_{b+a-t_1},\dots,{v}_{b+a-t_1+d_{1}-2})$ which is a shifted version of $\textbf{s}_{d_{1}}$. Let $h=b+a-t_1-1-d_{2}i$.
Due to	$b+a-t_1-1-(m+d_2-b)=2b-d_1-d_2-t_1-1$ and $	2b+t_2> d_1+d_2+t_1$,
one can get $h \geq-t_2$.
Recall that $b+a-t_1-1<m+d_2-d_1$,
based on \eqref{i_2-fulv} we have
\begin{equation}\label{sub_3-2}
	(\overline{v}_{b+a-t_1-1},v_{b+a-t_1},\dots,{v}_{b+a-t_1+d_{1}-2})=(\overline{v}_{h},\dots,{v}_{h+d_1-1}).
\end{equation}
Moreover, it is clear that $h<b-d_1$ by $b<d_1+d_2-t_2 \leq d_2+t_1$. Hence $-t_2\leq h<b-d_1$ implies that  $({v}_{h},\dots,{v}_{h+d_1-1})$ is a shift equivalent sequence of $\textbf{s}_{d_{1}}$ in $\mathcal{B}$ defined in \eqref{system_B}, a contradiction to \eqref{sub_3-2}.

{\bf Subcase (3.3)}: $b+a-t_{1}-1 = g$, which means $2b=d_1+d_2+t_1-t_2$. Thus we know
\begin{equation}\label{parameters_equal}
	({v}_{b+a-t_1},v_{b+a-t_1+1},\dots,{v}_{b+a-t_1+d_{1}-1})=(v_{-t_{2}},\dots,{v}_{-t_{2}+d_{1}-1}).	
\end{equation}
The equation in \eqref{parameters_equal} holds if and only if $d_1\,|\, (m+b-d_2)$ since the two subsequences in \eqref{parameters_equal} are aperiodic, where $m+b-d_2=(n-t_2)-(b+a-t_1)$.
Now the parameters satisfy
\begin{equation}\label{constraints}
	\left\{\begin{array}{cll}
		&t_{1}<b<d_{1}+d_{2}-t_{2}, \text{ i.e. }2k \leq b \leq 2k+4, \\
		&d_{2}\,|\, m+d_{2}-b, \\
		&2b=d_1+d_2+t_1-t_2, \\
		&d_1\,|\, m+b-d_2. \\
	\end{array}\right.  \ \
\end{equation}
In the following, we shall derive contradictions by \eqref{constraints} and the structure of sequences.
Since $m+b \geq (4k+2)+2k >4k+5 \geq d_1+d_2$ by \eqref{range_para}, we have $m+b-d_2>d_1$.
Moreover $d_1\,|\, m+b-d_2$, it follows $2d_1 \leq m+b-d_2$.
Thus by $b<d_{1}+d_{2}-t_{2}$,  we have $2d_1 \leq m+b-d_2 <m+d_{1}-t_{2}$, implying that $d_1 \leq m-t_2-1$.
Similarly, since $d_{2}\,|\, m+d_{2}-b$ and $m-b \geq (4k+2)-(2k+4)=2k-2>0$ by $k \geq 2$, we know that
$d_2 \leq m-b \leq m-t_1-1$ from \eqref{constraints}.

Hence $2b=d_1+d_2+t_1-t_2 \leq  (m-t_2-1)+(m-t_1-1)+t_1-t_2=2(m-t_2-1)$, that is, $t_1+1 \leq b \leq m-t_2-1 \leq m-t_1-2$, implying that $t_1 \leq \frac{m-3}{2}=2k-\frac{1}{2}$.
By $t_1\geq 2k-1$ in \eqref{range_para}, we can see that $t_1=2k-1$, which follows $2k \leq b \leq 2k+1$.
All in all, we have obtained that $t_1=2k-1$, $t_2\geq 2k$, $2k \leq b \leq 2k+1$, $d_1 \leq m-t_2-1\leq 2k+1$ and $d_2 \leq m-t_1-1\leq 2k+2$.
Finally, according to the value of $b$, we partition into two cases.

$b=2k+1$.
From $2b=d_1+d_2+t_1-t_2$ and $t_1=2k-1$, it follows $d_1+d_2-t_2=2k+3$.
Since $d_2 \leq 2k+2$, we have $d_1-t_2\geq 1$. Moreover  $d_1-t_2 \leq  (2k+1)-2k=1$, thus $d_1-t_2=1$,
implying $d_2 = 2k+2$. It contradicts that $d_2 \,|\, m-b$ with $m-b=(4k+2)-(2k+1)=2k+1$. 	

$b=2k$. From $2b=d_1+d_2+t_1-t_2$ and $t_1=2k-1$, it follows $d_1+d_2-t_2=2k+1$.
Due to $d_1\leq  2k+1$ and $d_2 \leq 2k+2$, we have $d_2-t_2 \geq 0$ and $d_1-t_2 \geq -1$, which indicates
$2k \leq t_2 \leq d_2  \leq 2k+2$ and $2k-1 \leq t_2-1 \leq d_1  \leq 2k+1$.
Because of $d_2 \,|\, m-b$ with $m-b=2k+2$, it is clear that $d_2=2k+2$.
Thus $m+b-d_2=2k$.
It follows from $d_1\,|\, m+b-d_2$ that $d_1=2k$.
Consequently, we have $a+b=(n-m-d_1)+b=4k+2$ and		
\begin{equation*}
	\left\{\begin{array}{cll}
		\sn	& \in  \mathcal{B}(8k+4, 4k+2,2k) \text{ with } add(\sn)=t_1=2k-1 \\
		R^{4k+2}(\sn)=\textbf{v}_n	& \in  \mathcal{B}(8k+4, 4k+2,2k+2) \text{ with } add(\sn)=t_2=2k+1 \\
	\end{array}\right.  \ \
\end{equation*}	
According to the structure of $\sn$, we know that
$s_{-2k}=\overline{s}_0$,  	$s_{i}=s_{i \, \text{ mod } 2k}$, $ -(2k-1)=-t_1\leq i\leq m+d_1-2=6k$ and 	$s_{6k+1}=\overline{s}_1$.
Thus
$$\textbf{v}_{d_2}=(s_{4k+2},\dots,s_{6k-1},s_{6k},s_{6k+1},s_{6k+2},s_{6k+3})=(s_{2},\dots,s_{2k-1},s_{0},\overline{s}_{1},s_{6k+2},s_{6k+3}).$$
Furthermore, we have
\begin{align*}
	\textbf{v}_{d_2}=&(s_{6k+4},\dots,s_{8k+3},s_{0},s_{1}) \\
	=&(s_{-2k},s_{-2k+1},\dots,s_{-2},s_{-1},s_{0},s_{1})=(\overline{s}_0,s_{1},\dots,s_{2k-3},s_{2k-2},s_{2k-1},s_{0},s_{1}).
\end{align*}		
Combine the two representations of $\textbf{v}_{d_2}$, it follows that
$\overline{s}_0=s_2=s_4=\dots=s_{2k-2}=s_0$	and ${s}_1=s_3=s_5=\dots=s_{2k-1}=\overline{s}_1$, which are contradictions.

Therefore, when $n=8k+4$ and $c\geq 6k+1$, there does not exist a pair of shift equivalent sequences $(\textbf{s}'_n, \textbf{v}'_n)$  in $\mathcal{B}(n, c)$ satisfying $add(\textbf{s}'_n)<add(\textbf{v}'_n)$.

On the other hand,
suppose a sequence has the following form
$$\sn=(s_0, \dots, s_{8k+3}) = (( \underline{\alpha\beta \alpha^{2k-2}} \,  \alpha\beta\alpha^{2k-2} \, \alpha\beta\alpha^{2k-2}\, \alpha^{2}) \alpha \beta\alpha^{2k}),$$
where $\beta = \overline{\alpha}$ and $\beta^l$ is the sequence given by $l$ repetitions of $\beta$. It is clear that
$\sn$ belongs to $\mathcal{B}(8k+4,4k+2, 2k)$  with $add(\sn)=2k-2$, where $\mathbf{s}_d=\mathbf{s}_{2k}$ is underlined.
Consider the sequences
$$\begin{array}{l}
	\mathbf{u}_n=R^{2k-2}(\sn) = ((\underline{\alpha ^{2k-1}\beta}\, \alpha^{2k-1}\beta \,  \alpha^{2k-1}\beta \,  \alpha^{2k}) \alpha\beta\alpha^{2}), \\
	\mathbf{v}_n=R^{6k}(\sn) = (( \underline{\alpha^{2k-3}\beta\alpha^4} \, \alpha^{2k-3}\beta\alpha^4\, \alpha^{2k-3}\beta\alpha^4 \,  \alpha^{2k-5}\beta)\alpha^2).
\end{array}
$$
It is clear that $\mathbf{u}_n \in \mathcal{B}(8k+4, 6k, 2k)$ with $add(\mathbf{u}_n) = 0$ and that
$\mathbf{v}_n \in \mathcal{B}(8k+4, 6k, 2k+2)$ with $add(\mathbf{v}_n) = 2$.
That is to say, for the sequence $\mathbf{u}_n$ in $\mathcal{B}(8k+4, 6k)$, its cyclic shift sequence $\mathbf{v}_n=R^{4k+2}(\mathbf{u}_n)$ belongs to $\mathcal{B}(8k+4, 6k)$ with larger $add(\mathbf{v}_n)$.
It indicates that the lower bound $c_0=6k+1$ is a tight bound such that any sequence $\sn\in \mathcal{B}(n, c)$ with $c\geq c_0$ is a sequence representative.


{\bf (iv)}	When $n=8k$, it remains to deal with the case of $t_{1}<b<d_{1}+d_{2}-t_{2}$, that is $2k-1<b < (m-t_1)+(m-t_1-1)-(t_1+1)<2k+1$, implying $b=2k$.
The proof of this case of $n=8k$ is similar as the proof of  $n=8k+4$.
Specially, {\bf Subcase (3.3)} can be reduced as ``$4k=2b=d_1+d_2+t_1-t_2$ implies that $d_1=2k+1$, $d_2=2k$ and  $t_2=2k$, which contradicts that $d_1 \, | \, (m+b-d_2)$.''
Therefore, when $n=8k$ and $c \geq 6k-1$, there does not exist a pair of shift equivalent sequences $(\textbf{s}'_n, \textbf{v}'_n)$  in $\mathcal{B}(n, c)$ satisfying $add(\textbf{s}'_n)<add(\textbf{v}'_n)$.

On the other hand,
suppose a sequence has the following form
$$\sn=(s_0, \dots, s_{8k-1}) = ((\underline{\alpha(\alpha\beta)^{k-1} }\, \alpha(\alpha\beta)^{k-1} \,\alpha(\alpha\beta)^{k-1} \,\alpha\beta) \,
\alpha(\alpha\beta)^{k}),$$
where $\beta = \overline{\alpha}$ and $(\alpha\beta)^l$ is the sequence given by $l$ repetitions of $\alpha\beta$, that is $\alpha\beta\alpha\beta\alpha\beta \dots\alpha\beta$. It is clear that
$\sn$ belongs to $\mathcal{B}(8k,4k, 2k-1)$ with $add(\sn)=2k-2$, where $\mathbf{s}_d=\mathbf{s}_{2k-1}$ is underlined.
Consider the sequences
$$\begin{array}{l}
	\mathbf{u}_n=R^{2k-2}(\sn) = ((\underline{(\alpha\beta)^{k-1} \alpha}\, (\alpha\beta)^{k-1} \alpha \,(\alpha\beta)^{k-1} \alpha \,  (\alpha\beta)^{k-1} \alpha \,  \beta) \alpha\alpha\beta), \\
	\mathbf{v}_n=R^{6k-2}(\sn) = (( \underline{(\alpha\beta)^{k-2} \alpha (\alpha\beta)^{2}} \, (\alpha\beta)^{k-2} \alpha (\alpha\beta)^{2} \, (\alpha\beta)^{k-2} \alpha (\alpha\beta)^{2} \, (\alpha\beta)^{k-3} \alpha \alpha)\beta).
\end{array}
$$
It is clear that $\mathbf{u}_n \in \mathcal{B}(8k, 6k-2, 2k-1)$ with $add(\mathbf{u}_n) = 0$ and that
$\mathbf{v}_n \in \mathcal{B}(8k, 6k-2, 2k+1)$ with $add(\mathbf{v}_n) = 2$.
That is to say, for the sequence $\mathbf{u}_n$ in $\mathcal{B}(8k, 6k-2)$, its cyclic shift sequence $\mathbf{v}_n=R^{4k}(\mathbf{u}_n)$ belongs to $\mathcal{B}(8k, 6k-2)$ with larger $add(\mathbf{v}_n)$.
It indicates that the lower bound $c_0=6k-1$ is a tight bound such that any sequence $\sn\in \mathcal{B}(n, c)$ with $c\geq c_0$ is a sequence representative.	
\hfill $\square$


\begin{thebibliography}{10}
	\bibitem{Trivium}
	C.~Canni$\rm{\grave{e}}$re and B.~Preneel.
	\newblock Trivium.
	\newblock {\em in New Stream Cipher Designs: The eSTREAM Finalists (Lecture
		Notes in Computer Science),} 4986: 244--266.
	Springer Berlin Heidelberg, 2008.
	
	
	\bibitem{Castellanos2022}
	A.~S. Castellanos, L.~Quoos, and G.~Tizziotti.
	\newblock Construction of sequences with high nonlinear complexity from a
	generalization of the hermitian function field.
	\newblock {\em 	J. Algebra Appl.}, 23(2): 2450037, Feb. 2024.
	

	
	\bibitem{ChenZChen}	
	Z. Chen, Z. Chen, J. Obrovsky  and A. Winterhof.
	\newblock	Maximum-order complexity and 2-adic complexity.
	{\em {IEEE} Trans. Inf. Theory}, 70(8): 6060--6067, Aug. 2024.
	
	
	
	\bibitem{Chenzhix}
	Z.~Chen, A.~I. G{\'{o}}mez, D.~G{\'{o}}mez-P{\'{e}}rez, and A.~Tirkel.
	\newblock Correlation measure, linear complexity and maximum order complexity for
	families of binary sequences.
	{\em Finite Fields Appl.},   78(101977): 1--11, Feb. 2022.
	
	
	

\bibitem{Ding}
C. Ding.
\newblock
Linear complexity of some generalized cyclotomic sequences.
{\em Int. J. Algebra Comput.}, 8(4): 431--442,  Aug. 1998.



\bibitem{DingLegendre}
C. Ding, T. Hesseseth, and W. Shan.
\newblock
On the linear complexity of Legendre sequences.
{\em {IEEE} Trans. Inf. Theory}, 44(3): 1276--1278, May 1998.



	
	\bibitem{Erdmann}
	D. Erdmann and S. Murphy.
	\newblock An approximate distribution for the maximum order complexity.
	\newblock {\em Des. Codes Cryptogr.}, 10(3): 325--339, Mar. 1997.
	
	
	\bibitem{Golomb2017}
	S.~W. Golomb.
	\newblock {\em Shift Register Sequences: Secure and Limited-Access Code
		Generators, Efficiency Code Generators, Prescribed Property Generators,
		Mathematical Models, 3rd ed.}
	\newblock World Scientific Singapore, 2017.
	
	
	\bibitem{Hell2008}
	M. Hell, T. Johansson, A. Maximov, and W. Meier.
	\newblock { The Grain family of stream ciphers}.
	\newblock { \em in New Stream Cipher Designs: The eSTREAM Finalists
		(Lecture Notes in Computer Science),}
	4986: 179--190.
	\newblock Springer Berlin Heidelberg, 2008.
	
	
	\bibitem{Helleseth}
	T. Helleseth.
	\newblock { Nonlinear shift registers - A survey and challenges}.
	in  
	\newblock {\em 	Algebraic Curves and Finite Fields: Cryptography and Other Applications,} 121--144.
	De Gruyter  Berlin, 2014.
	
	
\bibitem{LWin}
L.~I{\c{s}}{\i}k and A.~Winterhof.
\newblock	 Maximum-order complexity and correlation
measures. {\em Cryptography}, 1(1): 1--7, May 2017.

	
	
	
	
	\bibitem{JansenPhD}
	J.~A. Jansen.
	\newblock {\em Investigations on nonlinear streamcipher systems: Construction
		and evaluation methods}.
	\newblock { Ph.D. dissertation}, Delft Univ. Technol., Delft, The Netherlands, 1989.
	
	\bibitem{Jansen}
	J.~A. Jansen and D. Boekee.
	\newblock The shortest feedback shift register that can generate a given
	sequence.
	\newblock in {\em Advances in Cryptology - {CRYPTO} 1989},
	90--99. Springer New York, 1989.
	
	

	


\bibitem{JinL}	
L. Jin, L. Ma, C. Xing and R. Zhu.
\newblock A new family of binary sequences with low correlation via elliptic curves.
{\em {IEEE} Trans. Inf. Theory}, DOI: 10.1109/TIT.2025.3573103.



\bibitem{PKe2018}
P. Ke, Y. Jiang, and Z. Chen.
\newblock On the linear complexities of two classes of quaternary sequences of even length with optimal autocorrelation. {\em Adv. Math. Commun.}, 12(3): 525--539, Jul. 2018.


\bibitem{PKe2013}
P. Ke, J. Zhang, and S. Zhang.
\newblock On the linear complexity and the autocorrelation of generalized cyclotomic binary sequences of length $2p^m$.
{\em Des. Codes Cryptogr.},  67(3): 325--339, Jun. 2013.	


\bibitem{Nian}
N. Li and X. Tang.
\newblock On the linear complexity of binary sequences of period $4N$ with optimal autocorrelation value/magnitude. {\em IEEE Trans. Inf. Theory}, 57(11): 7597--7604,  Nov. 2011.


	
	\bibitem{liang}
	S. Liang, X. Zeng, Z. Xiao, and Z. Sun.
	\newblock Binary sequences with length $n$ and nonlinear complexity not less than
	$n/2$.
	\newblock {\em IEEE Trans. Inf. Theory}, 69(12): 8116--8125,
	Dec. 2023.
	
	
	\bibitem{Lidl1997}
	R.~Lidl and H.~Niederreiter.
	\newblock {\em Finite fields.
		\newblock in Encyclopedia of Mathematics and its
		Applications, 2nd~ed.} Cambridge University Press U.K., 1996.
	
	\bibitem{LKK2}
	K. Limniotis, N. Kolokotronis, and N. Kalouptsidis.
	\newblock On the nonlinear complexity and {Lempel-Ziv} complexity of finite
	length sequences.
	\newblock {\em {IEEE} Trans. Inf. Theory}, 53(11): 4293--4302,
	Nov. 2007.
	
	
	\bibitem{LXY}
	Y. Luo, C. Xing, and L. You.
	\newblock Construction of sequences with high nonlinear complexity from
	function fields.
	\newblock {\em {IEEE} Trans. Inf. Theory}, 63(12): 7646--7650,
	Dec. 2017.
	
	\bibitem{Massey1969}
	J.~Massey.
	\newblock Shift-register synthesis and {BCH} decoding.
	\newblock {\em {IEEE} Trans. Inf. Theory}, 15(1): 122--127, Jan.
	1969.
	
	
	\bibitem{MasseyShirlei1996_Crypto}
	J. L. Massey and S. Serconek.
	\newblock Linear complexity of periodic sequences: A general theory.
	\newblock in 
	{\em Advances in Cryptology - CRYPTO 1996},
	358--371. Springer Berlin Heidelberg. 1996.
	
	
	\bibitem{Niederreiter1999}
	H. Niederreiter.
	\newblock Some computable complexity measures for binary sequences.
	\newblock in {\em Sequences and their Applications 1999}, 67--78. Springer
	London, 1999.
	
	
	\bibitem{Niederreiter2003}
	H. Niederreiter.
	\newblock Linear complexity and related complexity measures for sequences.
	\newblock in {\em Progress in Cryptology - {INDOCRYPT} 2003}, 1--17.
	Springer Berlin Heidelberg, 2003.
	
	\bibitem{HX}
	H. Niederreiter and C. Xing.
	\newblock Sequences with high nonlinear complexity.
	\newblock {\em {IEEE} Trans. Inf. Theory}, 60(10): 6696--6701,
	Oct. 2014.
	
	
	\bibitem{LWC}
	Computer Securty Resoource~Center (NIST).
	\newblock Lightweight cryptography.
	\newblock {\em https://csrc.nist.gov/projects/lightweight-cryptography.html},
	2020.
	
	
	
	
	
	\bibitem{PM2006} G. Petrides and J. Mykkeltveit.
	\newblock On the classification of periodic binary sequences into nonlinear complexity classes. in
		\newblock {\em Sequences and Their Applications - SETA 2006.}  209--222.
	\newblock Springer Berlin, 2006.
	
	
	\bibitem{Petrides2008}
	G. Petrides and J. Mykkeltveit.
	\newblock Composition of recursions and nonlinear complexity of periodic binary
	sequences.
	\newblock {\em Des. Codes Cryptogr.}, 49(1-3): 251--264, Mar. 2008.
	
	\bibitem{PR}
	P.~Rizomiliotis. Constructing periodic binary sequences with maximum
	nonlinear span.  {\em {IEEE} Trans. Inf. Theory}, 52(9): 4257--4261, Sep. 2006.
	
	
	\bibitem{Rizomiliotis2005}
	P.~Rizomiliotis and N.~Kalouptsidis.
	\newblock Results on the nonlinear span of binary sequences.
	\newblock {\em {IEEE} Trans. Inf. Theory}, 51(4): 1555--1563,
	Apr. 2005.
	
	
	
	\bibitem{Rizo}
	P.~Rizomiliotis, N.~Kolokotronis, and N.~Kalouptsidis.
	\newblock On the quadratic span of binary sequences.
	\newblock {\em {IEEE} Trans. Inf. Theory}, 51(5): 1840--1848,
	May 2005.
	
	
	
	
	\bibitem{Rueppel1986}
	R. A.~Rueppel. {\em Analysis and Design of Stream Ciphers}.
	Springer Berlin, 1986.
	
	
	\bibitem{SZLH}
	Z. Sun, X. Zeng, C. Li, and T. Helleseth.
	\newblock Investigations on periodic sequences with maximum nonlinear
	complexity.
	\newblock {\em {IEEE} Trans. Inf. Theory}, 63(10): 6188--6198,
	Oct. 2017.
	
	
	
	\bibitem{Sunex}
	Z. Sun, X. Zeng, C. Li, Y. Zhang, and L. Yi.
	\newblock The expansion complexity
	of ultimately periodic sequences over finite fields.
	{\em {IEEE} Trans. Inf. Theory}, 67(11): 7550--7560, Nov. 2021.
	
	
	
	
	\bibitem{Tang2005}
	X. Tang, P. Udaya and P. Fan.
	\newblock A new family of nonbinary sequences with three-level correlation property and large linear span. {\em IEEE Trans. Inf. Theory}, 51(8): 2906--2914,  Aug. 2005.
	
	
	\bibitem{Xiao2018}
	Z. Xiao, X. Zeng, C. Li, and Y. Jiang.
	\newblock Binary sequences with period $N$ and nonlinear complexity $N -2$.
	\newblock {\em Cryptogr. Commun.}, 11(4): 735--757, Jul. 2019.
	
		
	
	\bibitem{XingC2003}
	C. Xing, P. V. Kumar and C. Ding.
	\newblock Low-correlation, large linear span sequences from function fields.
	{\em IEEE Trans. Inf. Theory}, 49(6): 1439--1446,  Jun. 2003.
	
	
	
	
	\bibitem{XingLam}
	C. Xing and K. Y. Lam.
	\newblock Sequences with almost perfect linear complexity profiles and curves over finite fields.
	{\em IEEE Trans. Inf. Theory}, 45(5): 1267--1270,  May 1999.
	
	
	
	\bibitem{YanKe}
	F. Yan, P. Ke, and Z. Chang.
	\newblock The symmetric 2-adic complexity of Tang-Gong interleaved sequences from Legendre sequence pair.
	{\em Cryptogr. Commun.},  17(1): 167--179, Feb. 2025.
	
	
	
	
	\bibitem{Yilin}
	L. Yi, X. Zeng, and Z. Sun.
	\newblock On finite length nonbinary sequences
	with large nonlinear complexity over the residue ring $\mathbb{Z}_m$.
	{\em Adv. Math.	Commun.}, 15(4): 701--720, Nov. 2021.
		
	
	
	
	\bibitem{yuan}
	Q.~Yuan, C.~Li, X.~Zeng, T.~Helleseth, and D.~He.
	\newblock Further investigations on nonlinear complexity of periodic binary
	sequences. {\em {IEEE} Trans. Inf. Theory}, 70(7): 5376--5391, Jul. 2024.
	
	


	
	









	
	
	
	
	
	
	%
	%
	%
	%
	
	
	
	%
	
	
\end{thebibliography}
\end{document}